\documentclass[a4paper, 12pt]{amsart} 
\usepackage{fullpage}
\usepackage[numbers,sort&compress,merge,super]{natbib}

\newcommand{\abs}[1]{\left \lvert #1 \right \rvert}
\newcommand\J{\mathrm{J}_{a,b}}
\newcommand\Lag{\mathrm{L}_b}
\newcommand{\av}[1]{\left \langle #1 \right \rangle}

\DeclareMathOperator{\tr}{Tr}

\theoremstyle{plain}
\newtheorem{theorem}{Theorem}[section]
\newtheorem{proposition}[theorem]{Proposition}
\newtheorem{lemma}[theorem]{Lemma}
\newtheorem{corollary}[theorem]{Corollary}
\theoremstyle{definition}
\newtheorem{remark}[theorem]{Remark}
\newtheorem{conjecture}[theorem]{Conjecture}
\newtheorem{example}[theorem]{Example}

 \theoremstyle{definition}

\begin{document}
\title{Moments of the transmission eigenvalues, proper delay times 
    and random matrix theory II}
\author{F.~Mezzadri}
\address{School of Mathematics, University of Bristol, Bristol BS8
 1TW, UK}
\email{f.mezzadri@bristol.ac.uk}

\author{N.~J.~Simm}
\address{School of Mathematical Sciences, Queen Mary, University of
  London, Mile End Road, London E1 4NS, UK§}
\email{n.simm@qmul.ac.uk}

 \thanks{Research partially supported by EPSRC, grant no: EP/G019843/1,
   and by the Leverhulme Trust, Research Fellowship no:
   RF/4/RFG/2009/0092.}

\begin{abstract}
  We systematically study the first three terms in the asymptotic
  expansions of the moments of the transmission eigenvalues and proper
  delay times as the number of quantum channels $n$ in the leads goes
  to infinity. The computations are based on the assumption that the
  Landauer-B\"{u}ttiker scattering matrix for chaotic ballistic
  cavities can be modelled by the circular ensembles of Random Matrix
  Theory (RMT).  The starting points are the finite-$n$ formulae that
  we recently discovered.\cite{MS11a}  Our analysis includes all the
  symmetry classes $\beta \in \{1,2,4\}$; in addition, it applies to
  the transmission eigenvalues of Andreev billiards, whose symmetry
  classes were classified by Zirnbauer~\cite{Zir96} and Altland and
  Zirnbauer.\cite{AZ97}  Where applicable, our results are in
  complete agreement with the semiclassical theory of mesoscopic
  systems developed. by Berkolaiko \textit{et al.}~\cite{BHN08} and
  Berkolaiko and Kuipers.\cite{BK10,BK11} Our approach also
  applies to the Selberg-like integrals. We calculate the first two
  terms in their asymptotic expansion explicitly.
\end{abstract}
\maketitle

\tableofcontents
\section{Introduction}
\label{se:introduction}
A quantum dot is often modelled by a two-dimensional billiard with
holes on the boundary, whose sizes are proportional to the number of
quantum channels in the leads.  Because of the mesoscopic dimensions
of the ballistic cavity, quantum mechanical phase coherence plays an
important role in the dynamics of an electron inside it.  Therefore,
the electric current has an intrinsic stochastic nature, whose
fluctuations are of theoretical and experimental
interest. Furthermore, whenever the classical limit of the dynamics is
chaotic, it is expected that such fluctuations should be characterized
by universal features that are well described by \emph{Random Matrix
  Theory} (RMT).\cite{BS88,BS90,Bee93,BM94,JPB94,Bee97}

At low temperatures and voltage the scattering inside the cavity is
elastic and is described by the Landauer-B\"uttiker scattering matrix
\begin{equation}
  \label{eq:landauer_buttiker}
  S := \begin{pmatrix} r_{m \times m} & t'_{m \times n} \\
                      t_{n \times m} & r'_{n \times n}
      \end{pmatrix},
\end{equation}
where $m$ and $n$ are the number of quantum channels in the left and
right leads respectively.  The sub-blocks $r_{m \times m}$, $t_{n
  \times m}$ and $r'_{n \times n}$ and $t'_{m \times n}$ are the
reflection and transmission matrices through the incoming and outgoing
lead. Without loss of generality we shall assume that $m >n$. The
dimensionless quantum conductance at zero temperature is given by
\begin{equation}
  \label{eq:conductance}
  \tr tt^\dagger = T_1 + \dotsb + T_n,
\end{equation}
$T_1,\dotsc,T_n$ are the eigenvalues of the
transmission matrix $tt^\dagger$.  Since $S$ is unitary,
$T_1,\dotsc,T_n$ lie in the interval $[0,1]$.  

The Wigner-Smith time delay matrix is defined by
\begin{equation}
  \label{eq:time_delay}
  Q := -i\hbar S^{-1}\frac{\partial S}{\partial E}.
\end{equation}
The eigenvalues $\tau_1,\dotsc,\tau_n$ of $Q$ are called
\emph{proper delay times}, and their average
\begin{equation}
  \label{eq:delay_time}
  \tau_{\mathrm{W}} := \frac{1}{n} \tr Q = \frac{1}{n}\left(\tau_1 +
    \dotsb + \tau_n\right)
\end{equation}
is the \emph{Wigner delay time}.  The number of proper delay times
$n$ is the total number of quantum channels in the leads.  As the name
suggests, $\tau_{\mathrm{W}}$ is a measure of the time that the
electron spends in the ballistic cavity.  The purpose of this paper is
to provide a comprehensive asymptotic study in the limit as $n\to
\infty$ of the moments of the densities of the transmission
eigenvalues and of the delay times in chaotic quantum dots.

In two pioneering letters Bl\"umel and Smilansky\cite{BS88,BS90}
discovered that if the classical limit of the dynamics of a scattering
system is chaotic, then the spectral correlations of the scattering
matrix are well described by those of matrices in one of the circular
Dyson's ensembles: the Circular Orthogonal Ensemble (COE), the
Circular Unitary Ensemble (CUE) and the Circular Symplectic Ensemble
(CSE). If the dynamics is time-reversal invariant and $K^2=1$, where
$K$ is the time-reversal operator, then the appropriate ensemble is
the COE; if $K^2=-1$ it is the CSE; if the system does not have any
symmetry, then the appropriate ensemble is the CUE.

Under the assumption that the quantum dynamics can be modelled by RMT,
the transmission eigenvalues have the joint probability density
function (\emph{j.p.d.f.}) (Refs.~\citenum{BM94,Bee97,DBB10,For06b} and~\citenum{JPB94}) 
\begin{equation}
 \label{eq:tr_eig_Aqd}
 p^{(\beta,\delta)}(T_1,\dotsc,T_n) := \frac{1}{C
 }\prod_{j=1}^{n}T_{j}^{\alpha}\left(1 - T_j\right)^{\delta/2} 
 \prod_{1 \leq j < k \leq n}\left.
   \lvert T_{k}-T_{j}\right \rvert^{\beta},
\end{equation}
where $\beta \in \left \{1,2,4\right\}$, $\alpha =
\frac{\beta}{2}\left(m-n +1\right) - 1$ and $\delta \in \left
  \{-1,0,1,2\right \}$.  Throughout this paper, $C$ refers to a
normalization constant that may change at each occurrence.  The
right-hand side of~\eqref{eq:tr_eig_Aqd} is the \emph{j.p.d.f.} of the
eigenvalues of matrices in the Jacobi ensembles:
\begin{equation}
  \label{eq:jacobi_density}
   p_{\J}(x_1,\dotsc,x_n) :=  \frac{1}{C}\prod_{j=1}^n x_j^{\beta/2 (b+1)-1}
  (1-x_j)^{\beta/2(a+1)-1}\prod_{1\le j < k \le n}\abs{x_k - x_j}^{\beta},
\end{equation}
where $x_j \in [0,1]$, $j=1,\dotsc,n$,
\begin{equation}
  \label{eq:jacobi_parameters}
  a = \frac{2}{\beta}\left(1 + \frac{\delta}{2}\right) - 1 \quad
  \text{and} \quad b = m - n.
\end{equation}

The parameter $\delta$ is different from zero only for \emph{Andreev
  quantum dots,} which are ballistic cavities in contact with a
superconductor.  For such systems the symmetry classes are not Dyson's
ensembles, but those associated to symmetric spaces, which were
classified by Zirnbauer,\cite{Zir96} Altland and
Zirnbauer,\cite{AZ97} and Due\~{n}ez.\cite{Due01,Due04}  These
symmetries are parametrized by the pair of integers
\begin{equation}
\label{eq:symmetries}
(\beta,\delta) \in \left \{(1,-1),(2,-1),(4,2),(2,1)\right \}.
\end{equation}

The moments of the density of the transmission eigenvalues are
\begin{equation}
 \label{trans_mom}
 M^{(\beta)}_{\mathrm{J}_{a,b}}(k,n) := 
\av{\tr \bigl [\left(tt^\dagger\right)^k\bigr]} = \av{T_{1}^{k}+\ldots+T_{n}^{k}},
\end{equation}
where the angle brackets denote the average taken with respect to the
\emph{j.p.d.f.}~\eqref{eq:jacobi_density}. 

% These moments are
% related to the cumulants
% $\av{\av{\kappa_j}}$ of the charge transmitted over a
% finite interval of time by the generating function \cite{LLY95} (see
% also appendix A in~\cite{BB00})
% \begin{equation}
%  \label{eq:gen_fun}
%  \sum_{j=1}^{\infty}\frac{x^{j}}{j!}\av{\av{\kappa_{j}}}
%  = -\sum_{k=1}^{\infty}\frac{(-1)^{k}}{k}  (e^{x}-1)^{k} 
%  \av{\tr \bigl [\left(tt^\dagger\right)^k\bigr]}.
% \end{equation}
% The charge cumulants can be measured directly in
% experiments~\cite{BGSLR05} and 

The assumption that the statistical fluctuations for the electric
current in quantum dots are modelled by RMT leads to an explicit
formula for the \emph{j.p.d.f.} of the proper delay times
too.\cite{BFB97} We have
\begin{equation}
\label{eq:jpdfdelaytime}
P_{\beta}(\gamma_{1},\ldots,\gamma_{n}) := \frac{1}{C} 
\prod_{j=1}^{n}\gamma_{j}^{ n\beta/2}e^{-\beta \tau_{\mathrm{H}} 
\gamma_{j}/2} \prod_{1 \leq j < k \leq n} \abs{\gamma_{k}-\gamma_{j}}^{\beta},
\end{equation}
where $\gamma_j = \tau_j^{-1}$ and $\tau_{\mathrm{H}}$ is the
Heisenberg time.  In our setting $\tau_{\mathrm{H}}=n$.  The
\emph{j.p.d.f.}~\eqref{eq:jpdfdelaytime} is the density of the
Laguerre ensemble, namely   
\begin{equation}
  \label{eq:laguerre_density}
  p_{\Lag}(x_1,\dotsc,x_n) :=  \frac{1}{C}\prod_{j=1}^n
  x^{\beta/2 (b+1)-1}_je^{-\beta x_j/2}\prod_{1\le j < k \le n}\abs{x_k - x_j}^{\beta},
\end{equation}
for $x_j \in [0,\infty)$, $j=1,\dotsc,n$.

From~\eqref{eq:jpdfdelaytime} the moments of the density of proper
delay times are the \emph{negative} moments of the density of the
Laguerre ensembles:
\begin{equation}
   \label{eq:mom_time_del}
   M^{\left(\beta\right)}_{\mathrm{L}_b}(-k,n) :=
   \frac{1}{n^k}\left \langle \tr Q^k\right \rangle = \frac{1}{n^k}
   \av{\gamma_{1}^{-k}+\ldots+\gamma_{n}^{-k}}, \quad k < n\beta/2 +
   1, 
\end{equation}
The above average is taken with respect to the
\textit{j.p.d.f.}~\eqref{eq:laguerre_density} and the parameter $b$ is
given by
\begin{equation}
\label{eq:b_neg_lag}
b=n-1 + 2/\beta.
\end{equation}
For convenience, we have used the same notation to denote the
parameters $b$ in the Jacobi and Laguerre ensembles.  In the rest of
the article, it will be clear from the context to which ensembles it
refers to.
  
We compute the first three terms of the asymptotic expansions of the
averages~\eqref{trans_mom} and~\eqref{eq:mom_time_del} in the limit as
$n \to \infty$. The starting points of our analysis are the finite-$n$
formulae that we computed in a previous publication.\cite{MS11a}
Explicit expressions for finite moments usually are rather involved
and not suitable for an asymptotic analysis.  However, our
results\cite{MS11a} have the advantage of allowing an asymptotic study
beyond the leading order.

Due to their applications to physics and multivariate analysis, the
leading order asymptotics of the moments of the eigenvalue density for
the Laguerre and Jacobi ensembles have been studied in detail both
within RMT and using semiclassical
techniques.\cite{BYK87,BB96,CC04,Col05,Nov07,BHN08,BK10,BK11,CDLV10,Kra10}
The finite-$n$ moments of the transmission eigenvalues for $\beta=2$
were computed by Novaes,\cite{Nov08} and Vivo and Vivo.\cite{VV08}
Besides the first part of this work,\cite{MS11a} recently another
paper appeared\cite{LV11} where the finite moments of the transmission
eigenvalues for $\beta \in \{1,2,4\}$ were computed using a different
approach.  The leading order term of the density of the proper delay
times was studied in Refs.~\citenum{SFS01,SS03} and~\citenum{SSS01}.
However, in these articles our assumptions of ideal physical settings
do not hold.

In a pioneering paper Johansson\cite{Joh98} looked at the
fluctuations of the linear statistics $\sum_{j}f(x_j)$ as $n \to
\infty$ for ensembles whose \emph{j.p.d.f.} is
\begin{equation*}
  p(x_1,\dotsc,x_n) = \frac{1}{C}\prod_{j=1}^n\exp\bigl(-nV(x_j)\bigr)
  \prod_{1\le j < k \le n}\abs{x_k-x_j}^\beta.
\end{equation*}
The potential $V(x)$ is a polynomial of even degree and $f(x)$ is a
suitable test function.  He proved that $\sum_j f(x_j)$ converges to a
normal random variable with finite variance.  In the same paper he
computed the next to leading order correction for arbitrary values of
$\beta$. Subsequently, Dumitriu and Edelman\cite{DE06} studied the
global fluctuations of the spectra of the $\beta$-Hermite and
$\beta$-Laguerre ensembles, that is ensembles in which $\beta$ can be
any positive real number. They also computed the next to leading order
correction to the Mar\v{c}enko-Pastur law.  Recent results on the
global fluctuations of the spectra of $\beta$-ensembles include
Refs.~\citenum{KS10} and~\citenum{BG11}.

The statistical properties of quantum transport have been studied with
semiclassical techniques too.  The semiclassical approach is
completely independent of a RMT analysis.  Although the RMT conjecture
for closed quantum systems with a chaotic classical limit has a long
history,\cite{BT77,BGS84,Ber85,SR01,MHBHA04,MHBHA05} semiclassical
techniques have been applied to quantum transport relatively
recently.\cite{RS02,HMBH06,BHMH06,BHN08,KS07,KS08,BK10,BK11,BK12,Nov11}
Within this framework, the elements of the scattering matrix are
approximated by sums over classical trajectories connecting incoming
and outgoing channels in the leads. What is most important is that the
outcome of semiclassical calculations are predictions for
\emph{individual} energy-averaged chaotic systems. This should be
distinguished from the RMT approach, which involves calculating an
average over an appropriate \emph{ensemble} of different systems. The
semiclassical limit $\hbar \to 0$ corresponds to the limit $n \to
\infty$ in RMT.

Semiclassical calculations are perturbative in nature.  The inverse
channel number $n^{-1}$ is proportional to $\hbar$ and is the small
parameter for the theory. Higher order terms are constructed by
including successively more intricate families of classical
trajectories into the sums and then performing the appropriate
combinatorics. This task was completed to all orders for the
conductance\cite{HMBH06} and for the shot noise.\cite{BHMH06} Higher
moments of the density of the transmission eigenvalues are also known
at leading order from semiclassical theory.\cite{BHN08} Berkolaiko and
Kuipers\cite{BK10,BK11} computed the leading order generating function
of the moments of the proper delay times and the first two subleading
corrections of the generating functions of the moments of both the
transmission eigenvalues and the proper delay times for $\beta=1$ and
$\beta=2$.  Recently, Berkolaiko and Kuipers\cite{BK12} and
independently Novaes\cite{Nov11} announced two distinct combinatorial
treatments of the correlations of the scattering trajectories that
imply exact agreement at all orders in $n^{-1}$ between the
semiclassical and RMT calculations of the moments of the transmission
eigenvalues. At present, however, it seems quite difficult to extract
explicit expressions for the asymptotic expansions from the
combinatorial formulae.  The consistency of RMT with the semiclassical
analysis of mesoscopic systems is a major success of these two
approaches.  Indeed, in closed systems we expect semiclassical
theories to be consistent with RMT only when they describe the local
fluctuations of the energy levels; the averages that we compute are
affected by the global fluctuations of the spectra of random matrices.
It is far from being obvious a priori that we should expect an exact
agreement at all orders.

Not many semiclassical results are available for Andreev quantum
dots. Adagideli and Beenakker\cite{AB02} and subsequently Kuipers
\textit{et al.}\cite{KWPBR10,KEBPWR11} developed a semiclassical
theory of the energy gap in the density of states of Andreev
billiards. Berkolaiko and Kuipers\cite{BK11} computed the first two
corrections to the density of states. The leading order term of the
density of states was previously computed using RMT by Melsen
\textit{et al.}\cite{MBFB96,MBFB97} using a perturbative expansion
developed by Brouwer and Beenakker,\cite{BB96} which is based on the
computation of integrals of matrix elements of matrices in the 
circular ensembles.

Since the parameter $b$ in the integrals~\eqref{trans_mom}
and~\eqref{eq:mom_time_del} depends on the number of channels in the
leads, we need to specify how the limit is taken.  In the
moments~\eqref{trans_mom} both $m$ and $n$ tend to infinity, but the
scaling parameter
\begin{equation}
  \label{eq:cond_scaling}
  u = \frac{m}{n}
\end{equation}
remains finite.  In other words, we fix $a$ but let $b=n(u-1)$ grow to
infinity.   For the Laguerre ensemble it will be convenient to
introduce a new variable $w$ by defining 
\begin{equation}
  \label{eq:lag_par}
  b = n(w-1) + 2/\beta - 1.
\end{equation}
The moments~\eqref{eq:mom_time_del} are then recovered by setting
$w=2$. A similar generalization was also considered in
Ref.~\citenum{BFP98}.

Equations.~\eqref{eq:jacobi_parameters} and~\eqref{eq:lag_par} contain
$n$-independent terms which arise due to physical reasons and are
rather unusual from a purely mathematical point of view.  Indeed, they
have non trivial effects on the subleading order terms of the
averages~\eqref{trans_mom} and~\eqref{eq:mom_time_del}. (See,
\textit{e.g.}, Remark~\ref{re:de_com}.) For this reason the moments of
the proper delay times cannot be extracted from the results of
Dumitriu and Edelman\cite{DE06} beyond the leading order.

When $\beta >0$ the density~\eqref{eq:jacobi_density} is often referred
to as the \emph{Selberg density} and the averages
\begin{equation}
\label{eq:selberg_like}
\frac{1}{C}\int_{0}^{1}\dotsi
\int_{0}^{1}\left(\sum_{j=1}^{n}x_{j}^{k}
\right)
 \prod_{j=1}^{n}x_{j}^{\beta/2(b+1)-1}(1-x_{j})^{\beta/2(a+1)-1} 
\prod_{1 \leq j < l \leq n}\left
   \lvert x_{l}-x_{j}\right \rvert^{\beta}dx_{1}\dotsm dx_{n}
\end{equation}
as \emph{Selberg-like} integrals (see, \emph{e.g.},
Refs.~\citenum{For10,FW08} and~\citenum{Meh04}). Their name originates
from the fact that by setting $k=0$ and $C=n$, the right-hand side
of~\eqref{eq:selberg_like} becomes Selberg's integral, which was
introduced for the first time in quantum transport by Savin and
Sommers.\cite{SS06} They computed the average of the shot noise and
the Fano factor for arbitrary $\beta$ nonperturbatively.
Subsequently, Sommers \textit{et al.}\cite{SWS07} and Savin \textit{et
  al.}\cite{SSW08} applied this approach to calculate moments and
cumulants up to the $4$-th order of linear and non linear statistics
of the electric current.  Novaes\cite{Nov08} used the Selberg-like
integrals to study the moments of the transmission eigenvalues for
$\beta=2$ nonperturbatively. Khoruzhenko \textit{et al.}\cite{KSS09}
combined such integrals with the theory of symmetric functions to
compute higher order moments and cumulants of the conductance and shot
noise.
 
One may ask what the asymptotic expansion of~\eqref{eq:selberg_like}
is if both $a,b \to \infty$.  This problem was studied in
Refs.~\citenum{CDLV10,Kra10} and~\citenum{Nov10} when both $a$ and
$b$ are proportional to $n$.  The scaling that they used was
\begin{equation}
\label{eq:ab_selberg}
a =(v-1)n \quad   \text{and} \quad b=(u-1)n,
\end{equation}
with $u$ and $v$ positive constants.  The leading order term is known
when $\beta=2$ (Refs.~\citenum{CDLV10} and~\citenum{Kra10}) as well as
for general $\beta$ (Ref.~\citenum{Nov10}). It turns out that at leading order
the integral~\eqref{eq:selberg_like} is independent of $\beta$.  We
compute the first subleading correction for $\beta \in \{1,2,4\}$.  We
also show that the leading order terms are consistent with the results
in Refs.~\citenum{Col05} and~\citenum{Nov10}.

The structure of the article is the following: in
Sec.~\ref{se:preliminaries_us} we outline some of the main ideas
behind the proofs; in Sec.~\ref{se:unitary_sym} we treat the
asymptotics of the moments for $\beta=2$ ensembles;
Sec.~\ref{se:orth_sym} is devoted to $\beta=1$ and $\beta=4$
ensembles; in Sec.~\ref{se:selberg_like} we study Selberg-like
integrals; Sec.~\ref{se:conclusions} concludes the paper with comments
and open problems.

\section*{Acknowledgements}
We would like to express our gratitude to Gregory Berkolaiko, Peter
Forrester, Jonathan Keating, Jack Kuipers and Dmitry Savin for
stimulating and helpful discussions.

\section{Preliminaries}
\label{se:preliminaries_us}

As remarked in the introduction, in the first part of this
work\cite{MS11a} we computed the averages~\eqref{trans_mom}
and~\eqref{eq:mom_time_del} for any number of open channels. Our
results were expressed in terms of finite sums involving binomial
coefficients and \textit{Pochhammer Symbols}, defined by the ratio
\begin{equation}
  \label{eq:poc_sym_def}
  (x)_{(n)} = \frac{\Gamma(x+n)}{\Gamma(x)}.  
\end{equation}
This simple fact has two important consequences.  Firstly, the
moments~\eqref{trans_mom} and~\eqref{eq:mom_time_del} admit the
asymptotic expansions in integer powers of $1/n$:
\begin{subequations}
  \label{asymptoticexpansions}
\begin{align}
\label{eq:asy_ex_tra}
n^{-1}M^{(\beta)}_{\J}(k,n) & \sim
\sum_{p=0}^{\infty}\mathcal{T}_{k,p}^{(\beta,\delta)}(u) n^{-p}, \\
\label{eq:asy_ex_td}
n^{k-1} M^{(\beta)}_{\mathrm{L}_{b}}(-k,n) & \sim \sum_{p=0}^{\infty}
\mathcal{D}_{k,p}^{(\beta)}(w)n^{-p}.
\end{align}
\end{subequations}
Here and in the rest of the paper whenever the parameters $a$ and $b$
appear, they will depend implicitly on $u$ or $w$ according
to~\eqref{eq:jacobi_parameters} and~\eqref{eq:lag_par}.

The second important feature of our results\cite{MS11a} is that
the finite-$n$ formulae for~\eqref{trans_mom}
and~\eqref{eq:mom_time_del} contain meromorphic functions that are
particularly suitable to an asymptotic analysis.  In order to
understand this point, it is useful to consider an example for
$\beta=2$.

\subsection{An Example for $\beta=2$}
\label{ss:examp}

Let us define
\begin{equation}
 \label{eq:difference_jacobi}
 \Delta M^{(2)}_{\J}(k,n) := M^{(2)}_{\J}(k,n) -
 M^{(2)}_{\J}(k+1,n).
\end{equation}
In Ref.~\citenum{MS11a} we found that
\begin{equation}
\label{transmissionexample}
\Delta M^{(2)}_{\J}(k,n) = \frac{1}{k}
\sum_{j=0}^{k}\binom{k}{j}\binom{k}{j-1}U^{a,b}_{n,k,j},
\end{equation}
where
\begin{equation}
\label{coefficientexample}
U^{a,b}_{n,k,j} =
\frac{(a+b+2n-2j+k+1)(a+b+n)_{(k-j+1)}(a+n-j+1)_{(j)}
(b+n)_{(k-j+1)}}{(a+b+2n-j)_{(k+2)}(a+b+2n-j+1)_{(k)}(n+1)_{(-j)}}
\end{equation}
The moments are given explicitly by
\begin{equation}
  \label{eq:expl_mom_j}
  M^{(2)}_{\J}(k,n) = M^{(2)}_{\J}(1,n)- \sum_{j=1}^{k-1} \frac{1}{j}
\sum_{i=0}^{j}\binom{j}{i}\binom{j}{i-1}U^{a,b}_{n,j,i},
\end{equation}
where
\begin{equation}
\label{eq:Aomoto_int}
  M^{(\beta)}_{\J}(1,n) =  \frac{n(b+n)}{a + b + 2n}  
\end{equation}
is a particular case of \textit{Aomoto's integral} (see, \textit{e.g.},
Ref.~\citenum{Meh04}, pp. 309--310).

We also found a formula for the negative moments in the Laguerre ensemble,
\begin{equation}
\label{wignerexample}
 M_{\mathrm{L}_b}^{(2)}(-k,n) = \frac{1}{k}
 \sum_{j=0}^{n-1}\binom{k+j}{k-1}
 \binom{k+j-1}{k-1}\frac{(b+n)_{(-k-j)}}{(1+n)_{(-j-1)}}.
\end{equation}

The moments of the eigenvalue densities in ensembles with orthogonal
and symplectic symmetries have a similar structure: the summands are
products of two factors. The first one is $n$-independent and consists
of binomial coefficients. 
% \begin{equation}
% \label{binomials}
% \frac{1}{k}\binom{k}{j}\binom{k}{j-1} \quad \text{or} \quad 
% \frac{1}{k}\binom{k+j-1}{k-1}\binom{k+j}{k-1}.
% \end{equation}
This feature will obviously persist to all orders in the asymptotic
expansions of the right-hand sides of
Eqs.~\eqref{transmissionexample} and~\eqref{wignerexample}. 

The coefficients
\begin{equation}
  \label{eq:narayana_numbers}
  N(k,j) := \frac{1}{k}\binom{k}{j}\binom{k}{j-1}
\end{equation}
appear frequently in enumerative combinatorics, where they are known
as \textit{Narayana numbers}.  The $n$-independent factor in the
summands of~\eqref{wignerexample} are obtained from the Narayana
numbers by extending them to negative integers:
\begin{equation}
  \label{e q:neg_narayana}
  N(-k,j)= \frac{1}{k}\binom{k+j}{k-1}
 \binom{k+j-1}{k-1}, \quad k > 0.
\end{equation}
This identity follows from the elementary relation
\begin{equation}
  \label{eq:binom_id}
  \binom{-k}{j} = (-1)^j \binom{k + j -1}{k-1}.
\end{equation}

The $n$-dependent part of the factors in the summands
of~\eqref{transmissionexample} and~\eqref{wignerexample} are ratios of
Gamma functions, whose asymptotic expansion is\cite{TE51}
\begin{equation}
\label{fullgammaratio}
\frac{\Gamma(z+\alpha)}{\Gamma(z+\beta)} \sim z^{\alpha-\beta}
\sum_{i=0}^{\infty}(-1)^{i}\frac{(\beta-\alpha)_{(i)}}{i!}
B_{i}^{(\alpha-\beta + 1)}(\alpha)z^{-i}, \quad z \to \infty,
\end{equation}  
where $B^{(\gamma)}_i(x)$ are the generalized Bernoulli polynomials.
Now, substituting the leading order of~\eqref{fullgammaratio} into,
say, the $j$-th term in the sum~\eqref{transmissionexample} leads to a
power of $n$ whose exponent is obtained by summing or subtracting the
subscripts of the Pochhammer symbols in the
ratio~\eqref{coefficientexample}:
\begin{equation}
\label{simpleadding}
n^{1+k-j+1+j+k-j+1-k-2-k+j}=n.
\end{equation}
Crucially the exponent is independent of $j$.  Surprisingly this
phenomenon happens at all orders of $n$ for the moments of both the
transmission eigenvalues and proper delay times in all the symmetry
classes.  Thus, at least in principle, to compute the asymptotic
series~\eqref{eq:asy_ex_tra} up to a given order $p$ we need to
substitute into~\eqref{transmissionexample} the first $p$ terms of the
expansion~\eqref{fullgammaratio}.

\begin{remark}
 In addition to the Narayana numbers we will often need the
 \textit{Narayana polynomials}
\begin{equation}
\label{eq:narayanapolynomial}
N_{k}(u):=\frac{1}{k}\sum_{j=1}^{k}\binom{k}{j}\binom{k}{j-1}u^{j}
\end{equation}
as well as their generating function (see,
\textit{e.g.}, Ref.~\citenum{BSS93})
\begin{equation}
\label{narayanagenerator}
\begin{split}
\rho(u,s) & := \sum_{k=1}^\infty N_k(u)s^k \\
& = \frac{1-s(u+1)-\sqrt{1-2s+s^{2}-2us-2us^{2}+u^{2}s^{2}}}{2s}
\end{split}
\end{equation}

Indeed our results have a distinct combinatorial flavour. It is
interesting to note that also the semiclassical approach to these
problems relies heavily on combinatorics, although of a different
kind.  From the RMT point of view, it is not obvious a priori why
combinatorics should play such an important role in these
calculations.
\end{remark}

\subsection{Generating Functions} The coefficients in the
series~\eqref{eq:asy_ex_tra} and~\eqref{eq:asy_ex_td} become
increasingly involved as we go higher in the order of the expansion.
Then, it becomes convenient to express our results in terms of
\emph{generating functions.}

Given the sets of coefficients $\mathcal{T}^{(\beta,\delta)}_{k,p}(u)$ and
$\mathcal{D}^{(\beta)}_{k,p}(w)$, $k=0,1,\dotsc$, at a given order $p$
in expansions~\eqref{eq:asy_ex_tra} and~\eqref{eq:asy_ex_td}, we
define the generating functions
\begin{subequations}
 \label{genfns}
\begin{align}
\label{eq:gen_fun_momteig}
\mathcal T^{(\beta,\delta)}_p(u,s) & :=
\sum_{k=1}^{\infty}\mathcal{T}_{k,p}^{(\beta,\delta)}(u)s^{k}, \\
\label{eq:gen_fun_momttd}
 \mathcal{D}^{(\beta)}_p(w,s) & := 
\sum_{k=1}^{\infty}\mathcal{D}_{k,p}^{(\beta)}(w)s^{k}.
\end{align}
\end{subequations}

If these series are convergent in a neighbourhood of the origin
$B_\epsilon$, then $\mathcal D^{(\beta)}_p(w,s)$ and $\mathcal
T^{(\beta,\delta)}_p(u,s)$ are analytic in $B_\epsilon$ and define the
moments uniquely.

As already appeared in the remarks about
formula~\eqref{transmissionexample}, our results on the moments of the
transmission eigenvalues are based on computing the differences
between consecutive moments.  It is therefore convenient to know their
generating function too, which is defined by   
\begin{equation}
\label{gendiffs}
\Delta \mathcal T^{(\beta,\delta)}_p(u,s)  := \sum_{k=1}^{\infty}
\Delta \mathcal{T}^{(\beta,\delta)}_{k,p}(u) s^{k}.
\end{equation}
The coefficients of the powers of $s$ in this formula are defined by
\begin{equation}
  \label{eq:diff_mom_coeff}
  \Delta \mathcal{T}_{k,p}^{(\beta,\delta)}(u) :=
  \mathcal{T}_{k,p}^{(\beta,\delta)}(u)- \mathcal{T}_{k+1,p}^{(\beta,\delta)}(u).
\end{equation}
The moments at a given order $p$ in the asymptotic 
expansion~\eqref{eq:asy_ex_tra}  are then obtained from the formula
\begin{equation}
  \label{eq:totmom}
  \mathcal{T}_{k,p}^{(\beta,\delta)}(u) = \mathcal{T}_{1,p}^{(\beta,\delta)}(u) -
  \sum_{j=1}^{k-1}\Delta \mathcal{T}^{(\beta,\delta)}_{j,p}(u).
\end{equation}
The first moments $\mathcal{T}_{1,p}^{(\beta,\delta)}(u)$ can be computed directly from the asymptotic expansion of Aomoto's
integral~\eqref{eq:Aomoto_int}, which gives
\begin{equation}
  \label{eq:aom_ex}
  \mathcal{T}_{1,p}^{(\beta,\delta)}(u) =  \frac{u}{(u+1)^{p+1}}
\left[1 - \frac{2}{\beta}\left(\frac{\delta}{2}+ 1 \right)\right]^p.
\end{equation}
An immediate consequence of the definition~\eqref{gendiffs} is the following functional relationship with~\eqref{eq:gen_fun_momteig}:
\begin{equation}
\label{diffscorrespondence}
\mathcal T^{(\beta,\delta)}_{p}(u,s) = \frac{s}{s-1}\left(\Delta 
\mathcal
T^{(\beta,\delta)}_{p}(u,s)-\mathcal{T}_{1,p}^{(\beta,\delta)}(u)\right). 
\end{equation}

\section{Unitary Symmetry}
\label{se:unitary_sym}

The techniques and tools that we use are broadly similar for all the
symmetry classes.  However, as is usual in RMT, the power $\beta$ in
the Vandermonde determinant of the
\emph{j.p.d.f.}'s~\eqref{eq:jacobi_density}
and~\eqref{eq:laguerre_density} affects the complexity of the
calculations.  Although when $\delta \neq 0$ in~\eqref{eq:tr_eig_Aqd}
the matrix ensembles that model Andreev billiards do not belong to
Dyson's symmetry classes, for the sake of simplicity we shall use the
usual terminology when referring to an ensemble characterized by a
given $\beta$, namely unitary, orthogonal and symplectic if $\beta=2$,
$\beta=1$ and $\beta=4$ respectively.

\subsection{Leading Order}
In Sec.~\ref{sse:exact_results}, Remark~\ref{re:beta_in}, we shall
prove that the leading order contributions are the same for all
symmetry classes. Therefore, we shall compute the limits
\[
\lim_{n\to \infty} n^{-1}M^{(\beta)}_{\J}(k,n) \quad \text{and} \quad
\lim_{n\to \infty}n^{k-1}M^{(\beta)}_{\Lag}(-k,n)
\]
only when $\beta=2$, since in this case the finite-$n$ formulae 
simplify considerably.
\begin{proposition}
 \label{pro:leading_order}
The leading order terms of the series~\eqref{eq:asy_ex_tra}
and~\eqref{eq:asy_ex_td} are independent of $\beta$ and $\delta$ and
are given by
 \begin{subequations}
\begin{align}
 \label{wignerdelayleading}
  \mathcal{D}_{k,0}^{(\beta)}(w) & = \begin{cases}
    \frac{1}{w-1} & \text{if $k=1$}\\
    \frac{(w-1)^{1-2k}}{k-1}\sum\limits_{j=1}^{k-1}
  \binom{k-1}{j}\binom{k-1}{j-1}w^{j}
& \text{if $k>1$}
  \end{cases} \\
 \label{2leadjacobileading}
 \mathcal{T}_{k,0}^{(\beta,\delta)}(u) & = \frac{u}{u+1} -  \sum_{j=1}^{k-1}
 \frac{1}{j}\sum_{i=1}^{j}\binom{j}{i}\binom{j}{
   i-1}\frac{u^{2i}}{(u+1)^{2j+1}}.
\end{align}
\end{subequations}
\end{proposition}
\begin{proof}
 We first prove \eqref{wignerdelayleading}. We have
\begin{equation}
\label{leadingestimatewig}
\frac{(b+n)_{(-k-j)}}{(1+n)_{(-1-j)}} =
\frac{(wn)_{(-k-j)}}{(1+n)_{(-1-j)}} = n^{1-k}w^{-k-j} + O(n^{-k}),
\quad  n \to \infty.
\end{equation}
Inserting the right-hand side of this equation
into~\eqref{wignerexample} gives
\begin{equation}
\label{wignerdelayinfseries}
 \mathcal{D}^{(\beta)}_{k,0}(w) =
 \frac{1}{k}\sum_{j=0}^{\infty}\binom{k+j-1}{k-1}\binom{k+j}{k-1}w^{-k-j}.
\end{equation}

The infinite series on the right-hand side of
\eqref{wignerdelayinfseries} is convergent for $|w|>1$. Indeed, since
the ratio of successive terms is a rational function of the summation
index $j$, it is a special case of a hypergeometric series. We find
\begin{equation}
\label{leadingwignerderivation}
\frac{1}{k}\sum_{j=0}^{\infty}\binom{k+j-1}{k-1}\binom{k+j}{k-1}w^{-k-j}
= w^{-k}{}_2F_1(k,k+1;2;w^{-1}),
\end{equation}
where the hypergeometric function ${}_2F_1(a,b;c;z)$ is defined in
Appendix~\ref{jp_hf}, Eq.~\eqref{hypergeomdef}. If $k=1$ this sum
is a geometric series which converges to $1/(w-1)$; if $w=2$ it is the
mean of the Wigner time delay.  When $k>1$ applying Euler's
transformation formula (see, \emph{e.g.}, Ref.~\citenum{AS72}, Eq.
(15.3.3))
\begin{equation*}
{}_2F_1(a,b;c;z) = (1-z)^{c-a-b}{}_2F_1(c-a,c-b;c;z)
\end{equation*}
converts the hypergeometric series in \eqref{leadingwignerderivation}
into the finite sum
\begin{equation*}
%\label{finallineleadingwig}
\begin{split}
\mathcal{D}_{k,0}^{(\beta)}(w) &=
\frac{w^{k-1}}{(w-1)^{2k-1}}\,{}_2F_1(2-k,1-k;2;w^{-1}) \\
&= \frac{1}{(w-1)^{2k-1}}
\sum_{j=0}^{k-1}
\frac{w^{j}}{(k-1)_{(j-k+1)}(k)_{(j-k+1)}(2)_{(k-j-1)}(k-j-1)!},
\end{split}
\end{equation*}
where we used that $(-x)_{(j)} = 1/(x+1)_{(-j)}$. Finally, the identity
\begin{equation}
 \label{eq:id_bin_poc}
 \frac{1}{(k-1)_{(j-k+1)}(k)_{(j-k+1)}(2)_{(k-j-1)}(k-j-1)!}=
 \frac{1}{k-1}\binom{k-1}{j}\binom{k-1}{j-1}
 \end{equation}
yields the limit~\eqref{wignerdelayleading}.

In order to prove~\eqref{2leadjacobileading} we first set $a=0$ and
$b=(u-1)n$ in \eqref{coefficientexample}:
\begin{equation}
\label{leadingestimatetrans}
\begin{split}
U^{a,b}_{n,k,j} &= \frac{\bigl((u+1)n-2j+k+1\bigr)
 (un)_{(k-j+1)}(n-j+1)_{(j)}(un)_{(k-j+1)}}{((u+1)n-j)_{(k+2)}
((u+1)n-j+1)_{(k)}(n+1)_{(-j)}}\\
& \sim n\frac{u^{2k-2j+2}}{(u+1)^{2k+1}}+O(1), \quad n \to \infty.
\end{split}
\end{equation}
Inserting this expression into (\ref{transmissionexample}) and
relabelling the summation index gives
\begin{equation}
\label{eq:limit_diff}
\lim_{n \to \infty} n^{-1}\Delta M^{(2)}_{\J}(k,n)=
  \frac{1}{k}\sum_{j=1}^{k}\binom{k}{j}\binom{k}{
   j-1}\frac{u^{2j}}{(u+1)^{2k+1}}.
\end{equation}
Finally, Eq.~\eqref{2leadjacobileading} follows by taking the
limit as $n \to \infty$ of the identity
\begin{equation}
 \label{eq:gen_mom}
n^{-1} M^{(2)}_{\J}(k,n) = n^{-1} M^{(2)}_{\J}(1,n) -
n^{-1}\sum_{j=1}^{k-1} \Delta M^{(2)}_{\J}(j,n),
\end{equation}
where $M^{(2)}_{\J}(1,n)$ is Aomoto's integral~\eqref{eq:Aomoto_int}.
\end{proof}

An immediate consequence of this proposition is the following.
\begin{corollary}
\label{co:lead_gen_fs}
  At leading order, the generating
  functions~\eqref{eq:gen_fun_momteig} and~\eqref{eq:gen_fun_momttd}
  are independent of $\beta$ and $\delta$ and are given by

   \begin{subequations}
    \begin{align}
        \mathcal{D}_{0}^{(\beta)}(w,s) &= \frac{w-1-s-\sqrt{(w-1)^{2}
       -2s(w+1)+s^{2}}}{2} \label{eq:gen_fun_w} \\
        \mathcal{T}_0^{(\beta,\delta)}(u,s) &= 
    \frac{1}{2}\left(\sqrt{1+\frac{4us}{(u+1)^2(1-s)}}-1\right)(u+1) \label{eq:lead_gen_fun_te}.
    \end{align}
  \end{subequations}

\end{corollary}
\begin{proof}
  By definition and using formula~\eqref{wignerdelayleading} we have
\begin{equation}
\begin{split}
\mathcal{D}^{(\beta)}_{0}(w,s)& = \frac{s}{w-1} +
\sum_{k=2}^{\infty}\frac{N_{k-1}(w)}{(w-1)^{2k-1}}s^k\\
& = \frac{s}{w-1} + \frac{s}{w-1}\rho\bigl((w,s/(w-1)^2\bigr),
\end{split}
\end{equation}
where $N_k(u)$ are the Narayana
polynomials~\eqref{eq:narayanapolynomial} and $\rho(u,s)$ is their
generating function.  Inserting formula~\eqref{narayanagenerator}
for $\rho(u,s)$ gives~\eqref{eq:gen_fun_w}.

We first compute the generating function $\Delta
\mathcal{T}_0^{(\beta,\delta)}(u,s)$. From the explicit expression of
the moments~\eqref{2leadjacobileading} we obtain
  \begin{equation}
    \label{eq:hg-nag}
    \Delta \mathcal{T}_0^{(\beta,\delta)}(u,s) =
    \frac{1}{u+1}\rho\left(u^2,s/(u+1)^2\right),
  \end{equation}
  where, as in the proof of Eq.~\eqref{eq:gen_fun_w}, $\rho(u,s)$
  is the generating function~\eqref{narayanagenerator}. Elementary
  manipulations give
   \begin{equation}
  \label{transgenleadingdiff}
  \Delta \mathcal{T}^{(\beta,\delta)}_0(u,s) = \frac{1}{2s}
  \left(\sqrt{1+\frac{4us}{(u+1)^2(1-s)}}-1\right)(u+1)(s-1)+\frac{u}{u+1}.
\end{equation}
Finally, Eq.~\eqref{eq:lead_gen_fun_te}
follows from~\eqref{diffscorrespondence} and~\eqref{eq:aom_ex}.
\end{proof}

\begin{remark}
 Using Selberg's integral Novaes\cite{Nov10} derived the formula
\begin{equation}
\mathcal{T}^{(\beta,0)}_{k,0}(u) =
\sum_{j=0}^{k-1}\binom{k-1}{j}\binom{2j}{j}
\frac{(-1)^{j}}{j+1}\frac{u^{j+1}}{(u+1)^{2j+1}}. \label{marcelasymptresult}
\end{equation}
Taking the differences of this sum and comparing to
(\ref{eq:limit_diff}) leads to the combinatorial identity
\begin{equation}
\label{combidentity1}
 \frac{1}{k}\sum_{j=0}^{k}\binom{k}{j}\binom{k}{j-1}u^{2j}
=\sum_{j=0}^{k}\binom{k-1}{j-1}\binom{2j}{j}
\frac{(-1)^{j+1}}{j+1}u^{j+1}(u+1)^{2k-2j} .
\end{equation}
This formula has already received attention in enumerative
combinatorics, where it was proved using generating functions by
Coker.\cite{Cok03} Subsequently, Chen \textit{et al.},\cite{CYY08}
who referred to it as Coker's \textit{second} problem, gave a
bijective proof.  Combining Novaes' expression
and~\eqref{2leadjacobileading} gives an independent proof
of~\eqref{combidentity1}.

The identity (\ref{combidentity1}) is intrinsically connected to
averages in the Jacobi ensemble.  Therefore, it seems natural to ask
whether similar identities should appear when integrating over other
ensembles, in particular the Laguerre ensembles. It is a
straightforward consequence of our exact results\cite{MS11a} (see
also Refs.~\citenum{Dum03} and~\citenum{HT03}) that
\begin{equation}
\label{mylaguerreleading}
\lim_{n \to \infty}n^{-k-1}M_{\mathrm{L}_{b}}^{(\beta)}(k,n)
= \frac{1}{k}\sum_{j=1}^{k}\binom{k}{j}\binom{k}{j-1}w^{j}.
\end{equation}
In the same paper Novaes also computed this limit using
Selberg's integral and showed that
\begin{equation}
\label{marcellaguerreleading}
\lim_{n \to \infty}n^{-k-1}M_{\mathrm{L}_{b}}^{(\beta)}(k,n)
= \sum_{j=0}^{\left \lfloor \frac{k-1}{2} \right \rfloor}
\frac{1}{j+1}\binom{2j}{j}\binom{k-1}{2j}w^{j+1}(1+w)^{k-2j-1}.
\end{equation}
The equivalence between the sums (\ref{mylaguerreleading}) and
(\ref{marcellaguerreleading}) appeared in Ref.~\citenum{CYY08} too,
where it was referred to as Coker's \textit{first} problem.

It is very interesting and surprising that these combinatorial
identities can be proved solely using techniques from RMT.
\end{remark}
\begin{remark}
  When $w=2$ the generating function~\eqref{eq:gen_fun_w} reduces to
  \begin{equation}
    \label{leadingwigdelay}
     \mathcal{D}^{(\beta)}_0(2,s) = \frac{1-s-\sqrt{1-6s+s^{2}}}{2}, 
 \end{equation}
 which was computed in Ref.~\citenum{BK10} using both periodic orbit
 theory and RMT.
\end{remark}
\begin{remark}
  The leading order generating function~\eqref{eq:lead_gen_fun_te}
  agrees with the semiclassical formula computed by Berkolaiko
  \textit{et al.}\cite{BHN08} Their notation is slightly different
  from ours.  Their main asymptotic parameter is the total number of
  quantum channels $N = n + m = (u+1)n$; they also introduce the
  variable $\xi = u/(u+1)^2$.  From our point of view, it is more
  convenient to use the dimension of the ensemble as asymptotic
  variable.
\end{remark}
\begin{remark}
 Instead of deriving exact formulae for the moments and extracting
 their asymptotics, one could instead study the asymptotic expansion
 for the \textit{mean eigenvalue density}
\begin{equation}
\label{meanevdensity}
\rho_{n}(x) := \left\langle \sum_{j=1}^{n}\delta(x-x_{j})\right \rangle.
\end{equation}
The moments are by definition $\int_{}x^{k}\rho_{n}(x)dx$.   For
example, if the average is taken with respect to the \textit{j.p.d.f}
of the Laguerre ensemble, then the limit
\begin{equation}
\label{marcenkopastur}
\lim_{n \to \infty}n^{-1}\rho_{n}(x)
= \frac{\sqrt{(x-\nu_{-})(\nu_{+}-x)}}{2\pi x}
\end{equation}
is the Mar\v{c}enko-Pastur law.\cite{MP67}  The support
of~\eqref{marcenkopastur} is $\nu_{\pm} = (\sqrt{w} \pm 1)^{2}$, where
$w$ is defined as in \eqref{eq:lag_par}.
%somewhere we need to say that this scaling is not the usual scaling for wishart matrices
This fact was exploited to obtain the leading contribution to the
density of the delay times.\cite{BFB97} The positive moments of this
density are known too,\cite{OP97,HT03} from which the corresponding
negative moments may be obtained from a substitution $x \to x^{-1}$ in
(\ref{marcenkopastur}). They are easily seen to be in agreement with
formula (\ref{wignerdelayleading}).
\end{remark}

\subsection{Beyond Leading Order}
\label{sse:unitary_corrections}
We now compute the next two terms in the asymptotic expansion of the
moments of the transmission eigenvalues and proper delay times when
$\beta=2$.

The starting point to compute higher order corrections to
formulae~\eqref{wignerdelayleading} and~\eqref{2leadjacobileading} is
to expand the $n$-dependent factor
\begin{equation}
 \label{eq:wig_time_del_fac}
 \frac{(b +n)_{(-k-j)}}{(1 + n)_{(-1-j)}}
\end{equation}
for the moments of the proper delay times and the
coefficients~\eqref{coefficientexample} for those of the transmission
eigenvalues.  For this purpose we need the first three terms in the
expansion~\eqref{fullgammaratio} (see, \textit{e.g.},
Ref.~\citenum{AS72}, Eq. (6.1.47)):
\begin{equation}
\label{subleadingpochhammer}
\begin{split}
\frac{\Gamma(z+\alpha)}{\Gamma(z+\beta)} & \sim
z^{\alpha-\beta}\left(1+\frac{1}{2z}(\alpha-\beta)(\alpha +\beta-1)
\right. \\
&\quad +  \left.\frac{1}{12z^{2}}\binom{\alpha - \beta}{2}
\bigl(3(\alpha+\beta-1)^{2}-(\alpha-\beta+1)\bigr)+O(z^{-3})\right),
\quad z \to \infty.
\end{split}
\end{equation}
The coefficients of $1/z$ and $1/z^2$ will give the terms of order
$O(n^{-1})$ and $O(n^{-2})$ respectively in the asymptotic
expansions~\eqref{eq:asy_ex_tra} and~\eqref{eq:asy_ex_td}.  Such terms
will be polynomials in $u$ (or $w$), whose coefficients will be the
Narayana numbers~\eqref{eq:narayana_numbers} or will closely resemble
them.

In the Wigner-Dyson symmetry classes there is no contribution at the
next to leading order when $\beta=2$. However, in Andreev quantum dots
($\delta \neq 0$)there is always a non-zero correction.

\begin{proposition}
\label{unitarytermsvanishing}
Let $\beta=2$. The next to leading order terms in the asymptotic
expansions~\eqref{eq:asy_ex_tra} and~\eqref{eq:asy_ex_td} are 
\begin{subequations}
\begin{align}
\label{neglagzero}
\mathcal{D}_{k,1}^{(2)}(w) &= 0, \\
\label{transfirstunitarycorr}
\Delta \mathcal{T}^{(2,\delta)}_{k,1}(u)&=
\frac{\delta}{2(u-1)}\left(\frac{u-1}{u+1}\right)^{2k+2}
\mathcal{T}^{(\beta,\delta)}_{k+1,0}(-u),
\end{align}
\end{subequations}
where $\mathcal{T}^{(\beta,\delta)}_{k,0}(u)$ is independent of
$\beta$ and $\delta$ and is given explicitly
in~\eqref{2leadjacobileading}.
\end{proposition}
\begin{proof}
  For the negative Laguerre moments, we set $b=(w-1)n$ and insert the
  term of order $O(n^{-k})$ in equation \eqref{leadingestimatewig}
  into \eqref{wignerexample} leading to
\begin{align}
\label{easysimpwig}
\mathcal{D}_{k,1}^{(2)}(w) &= \frac{1}{2k}\sum_{j=0}^{\infty}
\binom{k+j-1}{k-1}\binom{k+j}{k-1}w^{-k-j-1}\bigl((k+j)(k+j+1)-wj(1+j)
\bigr) \\
\label{easysimpwigl}
&=\frac{k}{2}\sum_{j=0}^{\infty}\binom{k+j}{k}\binom{k+j+1}{k}w^{-k-j-1}
\tag{61'}\\
 & \quad
 -\frac{k}{2}\sum_{j=0}^{\infty}\binom{k+j-1}{k}\binom{k+j}{k}w^{-k-j}=0 \notag 
\end{align}
Equation~\eqref{neglagzero} follows from shifting the summation index
$j \to j-1$ in the sum~\eqref{easysimpwigl}.

We now prove~\eqref{transfirstunitarycorr}. We first use
\eqref{subleadingpochhammer} to expand \eqref{transmissionexample} in
an asymptotic series and obtain 
\begin{equation}
\label{easysimptra}
\Delta \mathcal{T}^{(2,\delta)}_{k,1}(u) = \frac{1}{k}\sum_{j=0}^{k}
\binom{k}{j}\binom{k}{j-1}
\frac{u^{2k-2j+1}}{(u+1)^{2k+2}}F^{(2,\delta)}_{k,1}(u,j)
\end{equation}
where 
\[
F^{(2,\delta)}_{k,1}(u,j) = 
j(\delta/2+1-j)u^2-uk\delta/2+(k-j+1)(k-j+\delta/2).
\]
Collecting the coefficients of $\delta$ in~\eqref{easysimptra} gives
\begin{equation}
\label{eq:delta_coef}
\begin{split}
  \Delta \mathcal{T}^{(2,\delta)}_{k,1}(u) &=  
\frac{\delta}{2k}\sum_{j=0}^{k}\binom{k}{j}\binom{k}{j-1}
\frac{u^{2k-2j+1}}{(u+1)^{2k+2}}\bigl((k-j+1)-uk+u^{2}j\bigr) \\
  &\quad +\frac{1}{k}\sum_{j=0}^{k}\binom{k}{j}\binom{k}{j-1}
 \frac{u^{2k-2j+1}}{(u+1)^{2k+2}}\bigl((k-j+1)(k-j)+u^{2}j(1-j)\bigr).
\end{split}
\end{equation}
Elementary manipulations show that the second sum in the right-hand
side of \eqref{eq:delta_coef}
vanishes, while the first simplifies to 
\begin{equation}
\label{floorceilingpoly}
\begin{split}
\Delta \mathcal{T}^{(2,\delta)}_{k,1}(u) &= 
\frac{\delta}{2(u+1)^{2k+2}}\left(\sum_{j=0}^{k}
\binom{k}{j}^{2}u^{2j+1}-\sum_{j=0}^{k}\binom{k}{j}\binom{k}{j+1}
u^{2j+2}\right)  \\
&=-\frac{\delta}{2(u+1)^{2k+2}}\sum_{j=0}^{2k}\binom{k}{\lfloor j/2 \rfloor}
\binom{k}{\lceil j/2 \rceil}(-u)^{j+1}. 
\end{split}
\end{equation}

Finally, the statement of the proposition follows from the
combinatorial identity (see Ref.~\citenum{CDLV10}, Proposition V7)
\begin{equation}
\label{eq:vivocombidentity}
\frac{u}{(u+1)^{2k+1}}\sum_{j=0}^{2k}
\binom{k}{\lfloor j/2 \rfloor}\binom{k}{\lceil j/2 \rceil}u^{j} 
= \sum_{j=0}^{k}\binom{k}{j}\binom{2j}{j}
\frac{(-1)^{j}}{j+1}\frac{u^{j+1}}{(u+1)^{2j+1}},
\end{equation}
which allows the polynomial in (\ref{floorceilingpoly}) to be
expressed in terms of the leading order term
$\mathcal{T}^{(\beta,\delta)}_{k,0}(u)$ given by
(\ref{marcelasymptresult}).
\end{proof}
\begin{remark}
  \label{re:de_com}
  At first it is quite surprising that the next to leading order
  term~\eqref{transfirstunitarycorr} is different from zero.  The fact
  that the correction at order $O(1/n)$ of the density of the
  eigenvalues of $\beta=2$ ensembles is usually zero is a general
  phenomenon due to the symmetry of the ensembles.  It was proved by
  Johansson\cite{Joh98} for ensembles of random matrices with an even 
  degree polynomial potential. It also appeared in the work of Dumitriu 
  and Edelman\cite{DE06} for the $\beta$-Hermite and
  $\beta$-Laguerre ensembles.  Dumitriu and Edelman\cite{DE06} showed
  that for these ensembles, the term of order $O(1/n^{p})$ is
  multiplied by a \textit{palindromic} polynomial in $-2/\beta$ that
  is independent of $n$. At the next to leading order this polynomial
  is zero at $\beta=2$.  For the Jacobi ensembles
  this property is still unproved, but given that it is a direct
  consequence of the symmetries of the Jack polynomials, one would
  expect the same behaviour.

  Why, then, is formula~\eqref{transfirstunitarycorr} not zero?  The
  answer is that for physical reasons the exponent of the factors
  $(1-x_j)$ in the \textit{j.p.d.f.} of the eigenvalues of the Jacobi
  ensembles \textit{does not scale} with $n$ (see
  Eq.~\eqref{eq:jacobi_parameters}).  If we adopt the
  scaling~\eqref{eq:ab_selberg}, then the next to leading order term
  is zero for $\beta=2$.  Indeed, this is the content of
  Proposition~\ref{pr:ntlsi} for the Selberg-like integrals.
\end{remark}

\begin{remark} For higher order terms, it is no longer possible to
  simplify the expressions which would appear in place of
  Eqs.~\eqref{easysimpwig} and~\eqref{easysimptra} through elementary
  manipulation. Instead, one requires a more systematic procedure to
  handle more complex polynomials, which become increasingly involved
  as we go further in the asymptotic expansion. However, as we shall
  see, often the coefficients in such polynomials are products of two
  binomial coefficients similar to those appearing in the first line
  of~\eqref{floorceilingpoly}. Thus, in many instances they can be
  conveniently expressed in terms of Jacobi polynomials (see
  Definition~\ref{jacdef}), whose properties can be used to obtain
  manageable formulae.  A first example of this procedure is given in
  the proof of Proposition~\ref{delaysubleadingtwoprop}.
\end{remark}

The next to leading order generating functions will be discussed in
Sec.~\ref{se:orth_sym}, where we shall derive a formula for general
$\delta$ with $\beta \in \{1,2,4\}$.

\begin{proposition}
\label{delaysubleadingtwoprop}
For $\beta=2$ the third coefficient in the
expansion~\eqref{eq:asy_ex_td} of the moments of the proper delay times is
\begin{equation}
\label{delaysubleadingtwo}
\mathcal{D}_{k,2}^{(2)}(w) = \frac{(k+1)(k+2)w}{12(w-1)^{k+3}}
P^{(2,2)}_{k-2}(\tilde{w}),
\end{equation}
where $P^{(\alpha,\beta)}_j(x)$ refers to the Jacobi polynomial of
degree $j$ and parameters $\alpha$ and $\beta$ defined in
Eq.~\eqref{jacdef} and
\begin{equation}
  \label{eq:tildew}
  \tilde{w} := \frac{w +1}{w-1}.
\end{equation}
\end{proposition}
\begin{proof}
  As in previous proofs, we begin by inserting
  \eqref{subleadingpochhammer} into~\eqref{eq:wig_time_del_fac} and
  extract the component of order $O(n^{-k-1})$.  Then, by
  substituting it into~\eqref{wignerexample} we arrive at
\begin{equation}
\label{p2dksubsub}
\mathcal{D}_{k,2}^{(2)}(w) =\frac{1}{k}
\sum_{j=0}^{\infty}\binom{k+j}{k-1}
\binom{k+j-1}{ k-1}w^{-k-j-2}G_{k,2}^{(2)}(w,j),
\end{equation}
where
\begin{equation*}
\begin{split}
G^{(2)}_{k,2}(w,j) &= \frac{w^{2}}{24}j(j-1)(j+1)(3j+2)
-\frac{w}{4}(k+j)(k+j+1)j(1+j) \\
& \quad +\frac{1}{24}(k+j)(k+j+1)(k+j+2)(3k+3j+1).
\end{split}
\end{equation*}
% where $G_{k,2}^{(2)}(w,j)$ is a polynomial of degree $2$ in $w$; its
% explicit definition is given in appendix~\ref{ap:add_exs},
% Eq.~\eqref{p2delayunitsubsub}. 

The demonstration now proceeds with lengthy and systematic, although
elementary, algebraic manipulations.  It is useful to give an example
that clarifies the pattern of the calculations. Consider the
contribution to \eqref{p2dksubsub} coming from the linear term 
in $G_{k,2}^{(2)}(w,j)$:
\begin{multline}
\label{polyexample2k_1}
-\frac{1}{4k}\sum_{j=0}^{\infty}\binom{k+j}{k-1}\binom{k+j-1}{k-1}
w^{-k-j-1}(k+j)(k+j+1)j(1+j)  \\
=-\frac{k(1+k)^{2}}{4}\sum_{j=0}^{\infty}\binom{k+j}{k+1}
\binom{k+j+1}{k+1}w^{-k-j-1}.
\end{multline}
Using the definition of hypergeometric function~\eqref{hypergeomdef}
and the relation~\eqref{wignerhypergeom} between hypergeometric
functions and Jacobi polynomials, the right-hand side
of~\eqref{polyexample2k_1} becomes
\begin{equation}
\label{polyexampled2k}
-\frac{k(k+1)^{2}(k+2)}{4}w^{-k-2} \left({}_{2}F_{1}(k+2,k+3,2,w^{-1})\right)
=-\frac{k(k+1)(k+2)wP^{(1,1)}_{k}(\tilde{w})}{4(w-1)^{k+3}}. 
\end{equation}

Repeating the same procedure for the other powers of $w$ in
$G_{k,2}^{(2)}(w,j)$ eventually gives
\begin{equation}
\label{jacobicombination}
\begin{split}
\frac{24(w-1)^{k+3}}{w(k+1)(k+2)}\mathcal{D}_{k,2}^{(2)}(w) &= 3k
\bigl(P^{(0,2)}_{k}(\tilde{w})
+P^{(2,0)}_{k}(\tilde{w})\bigr)-6kP^{(1,1)}_{k}(\tilde{w})  \\
&\quad +(3k-2)(P^{(1,2)}_{k-1}\bigl(\tilde{w})-P^{(2,1)}
_{k-1}(\tilde{w})\bigr).
\end{split}
\end{equation}
Let us subtract the right-hand side of
Eq.~\eqref{delaysubleadingtwo} from \eqref{jacobicombination};
then, the statement of the proposition is equivalent to the identity
\begin{multline}
\label{showiszero}
3k\bigl(P^{(0,2)}_{k}(\tilde{w})+P^{(2,0)}_{k}(\tilde{w})\bigr)
-6kP^{(1,1)}_{k}(\tilde{w}) \\
+(3k-2)\bigl(P^{(1,2)}_{k-1}(\tilde{w})
-P^{(2,1)}_{k-1}(\tilde{w})\bigr)-2P^{(2,2)}_{k-2}(\tilde{w})=0.
\end{multline}
The proof follows from the orthogonality of the Jacobi polynomials and
the following general procedure:
\begin{enumerate}
\item write each Jacobi polynomial in \eqref{showiszero} in terms of
  the polynomials $P^{(2,2)}_{j}(\tilde{w})$ using the connection
  formulae~\eqref{connection1};
\item apply the three-term recurrence relation \eqref{threetermrec}
  to each Jacobi polynomial until only the polynomials
  $P^{(2,2)}_{k}(\tilde{w})$ and $P^{(2,2)}_{k-1}(\tilde{w})$ appear in
  the formula;
\item the resulting expression is of the form
\[
  A_{k}(w)P^{(2,2)}_{k}(\tilde{w})+B_{k}(w)P^{(2,2)}_{k-1}(\tilde{w}),
\]
  for some rational functions $A_{k}(w)$ and $B_{k}(w)$;
\item a direct computation shows
\[
A_{k}(w) = B_{k}(w) =  0.
\]
\end{enumerate}
\end{proof}

\begin{remark}
\label{linindepremark}
Note that step (1) of the above proof amounts to writing each Jacobi
polynomial in a consistent linearly independent basis. By
formula~\eqref{connection1} the number of non-zero coefficients in
this basis is typically quite small and independent of $k$. This
property is essential for our procedure to work. Indeed, consider the
identities
\begin{align*}
P^{(0,2)}_{k}(\tilde{w})+P^{(2,0)}_{k}(\tilde{w}) &= 
\frac{(k+3)(k+4)P^{(2,2)}_{k}(\tilde{w})}{(2k+3)(k+2)}
+\frac{(k+2)P^{(2,2)}_{k-2}(\tilde{w})}{2k+3}\\
P^{(1,2)}_{k-1}(\tilde{w})-P^{(2,1)}_{k-1}(\tilde{w}) &= 
-P^{(2,2)}_{k-2}(\tilde{w})\\
P^{(1,1)}_{k}(\tilde{w}) &= \frac{1}{2}
\left(\frac{(k+3)(k+4)P^{(2,2)}_{k}(\tilde{w})}{(2k+3)(k+2)}
-\frac{(k+1)P^{(2,2)}_{k-2}(\tilde{w})}{2k+3}\right),
\end{align*}
which are a direct consequence of~\eqref{connection1}.  When 
substituted into (\ref{showiszero}), they immediately yield the
desired cancellations. 

In some of the proofs of this paper (more specifically in Lemma
\ref{unitarygenfntrans} and Propositions \ref{unitarysubsubbeta2prop},
\ref{delaysubsub1prop}, \ref{transsubsub1prop} and \ref{pr:ntlsi}) we
will use this method to justify the equivalence of two apparently
different combinations of Jacobi polynomials. This task is most
conveniently performed using a computer algebra package such as Maple
or Mathematica, whereby the rather heavy algebra involved in steps
(1)--(4) may be executed algorithmically.
\end{remark}

In order to keep the notation as simple as possible, we now introduce
a convention that we shall often use when proving statements about
generating functions.  If $p(x)$ is a polynomial or a function
analytic in a neighbourhood of $x=0$, then we shall write
$\bigl[x^k\bigr]p(x)$ to indicate the Taylor coefficient of $x^k$ in
$p(x)$.  This formalism can be trivially extended to a function of
several variables.  Now, when in a proof we write $\bigl[s^k\bigr]
\mathcal{T}_{p}^{(\beta,\delta)}(u,s)$ or
$\bigl[s^k\bigr]\mathcal{D}^{(\beta)}_p(w,s)$, we do not refer to
$\mathcal{T}^{(\beta,\delta)}_{k,p}(u)$ or
$\mathcal{D}^{(\beta)}_{k,p}(w)$, but to the Taylor coefficients of a
function, usually an algebraic function, whose explicit expression is
the statement we intend to prove, like, for example, the right-hand
sides of Eqs.~\eqref{eq:gen_fun_w} and~\eqref{eq:lead_gen_fun_te}
in Corollary~\ref{co:lead_gen_fs}.  In other words, adopting this
slight abuse of notation, Corollary~\ref{co:lead_gen_fs} is equivalent
to
\begin{equation*}
  \bigl[s^k\bigr]\mathcal{D}^{(\beta)}_0(w,s) =
  \mathcal{D}^{(\beta)}_{k,0}(w) \quad \text{and} \quad
  \bigl[s^k\bigr]\mathcal{T}_0^{(\beta,\delta)} (u,s) =
  \mathcal{T}^{(\beta,\delta)}_{k,0}(u).
\end{equation*}

We now have
\begin{corollary}
\label{lemma1wiggen}
The generating function of the moments \eqref{delaysubleadingtwo} is given by
\begin{equation}
\label{beta2p2delaylemma}
\mathcal D_{2}^{(2)}(w,s) = \frac{ws^{2}}%
{\bigl(s^{2}-2(w+1)s+(w-1)^{2}\bigr)^{5/2}}.
\end{equation}
\end{corollary}
\begin{proof}
  The idea of proof is straightforward. We first show that the Taylor
  coefficients of (\ref{beta2p2delaylemma}) satisfy a particular three
  term recurrence relation. Then, we demonstrate that the moments
  (\ref{delaysubleadingtwo}) satisfy exactly the same recurrence
  relation.  Finally, we are left to verify the initial conditions
  \begin{align*}
  \bigl[s\bigr]\mathcal D_{2}^{(2)}(w,s) &= 0,\\
  \bigl[s^{2}\bigr] \mathcal D_{2}^{(2)}(w,s) &= \frac{w}{(w-1)^{5}},
  \end{align*}
  which can be easily checked. 

  Differentiating (\ref{beta2p2delaylemma}) once with respect to $s$
  leads to the differential equation 
\begin{multline*}
  s\bigl(s^{2}-2(w+1)s+(w-1)^{2}\bigr)\frac{d\mathcal{D}_{2}^{(2)}(w,s)}{ds}\\
  +\bigl(5s(s-w-1)+4(w+1)s-2(w-1)^{2}-2s^{2}\bigr)\mathcal{D}_2^{(2)}(w,s)=0.
\end{multline*}
Inserting the power series expansion 
\[
\mathcal D_2^{(2)}(w,s) = \sum_{k=1}^{\infty}p_{k}(w)s^{k}
\]
 and equating coefficients of $s^{k}$ gives
\begin{equation}
\label{delaysubleadingtworec}
(w-1)^{2}(k-2)p_{k}(w)-(w+1)(2k-1)p_{k-1}(w)+(k+1)p_{k-2}(w)=0.
\end{equation}

We now have to prove that the moments (\ref{delaysubleadingtwo})
satisfy the recurrence relation
(\ref{delaysubleadingtworec}). Inserting (\ref{delaysubleadingtwo})
into the left-hand side of (\ref{delaysubleadingtworec}) and
simplifying leads to
\begin{equation*}
(k+2)(k-2)P^{(2,2)}_{k-2}(\tilde{w})-(2k-1)k\tilde{w}
P^{(2,2)}_{k-3}(\tilde{w})+(k-1)kP^{(2,2)}_{k-4}(\tilde{w}),
\end{equation*}
which is just the three term recurrence relation (\ref{threetermrec})
for Jacobi polynomials with $\alpha=\beta=2$ and $n=k-3$.
\end{proof}

As we go beyond the next to leading order in the
expansion~\eqref{eq:asy_ex_tra}, the corrections to the moments of the
density of the transmission eigenvalues become increasingly involved.
Therefore, although we need the explicit expression of such
coefficients in the proof, it is more convenient to state our results
only in terms of generating functions.

We first need the following lemma.

\begin{lemma}
\label{unitarygenfntrans}
The Taylor coefficients of
\begin{equation}
\label{transgensubsub2diff}
f(s;u)= \frac{u^{2}s}{(1-s)^{1/2}\bigl((u+1)^{2}-s(u-1)^{2}\bigr)^{5/2}} 
\end{equation}
as a function of $s$ are
\begin{equation}
  \label{eq:coef_tay}
\begin{split}
\bigl[s^{k}\bigr]f(s;u)& =u^{2}\frac{k(k+1)(u-1)^{k-2}}{(u+1)^{k+5}}
\biggl(\bigl(1-u^{2}\bigr)\bigl((6k)^{-1}P^{(1,1)}_{k-1}(\tilde{u}) \\
& \quad - \frac{2}{3} P^{(0,0)}_{k-1}(\tilde{u})\bigr)
-\frac{2u}{3}P^{(1,1)}_{k-2}(\tilde{u})\biggr),
\end{split}
\end{equation}
where
\begin{equation}
  \label{eq:u_tilde}
  \tilde u := \frac{u^2+1}{u^2-1}.
\end{equation}
\end{lemma}
\begin{proof}
  The idea of the proof is the same as that one of
  Corollary~\ref{lemma1wiggen} and consists of two parts. First we
  obtain a recurrence relation for the Taylor coefficients of the
  generating function (\ref{transgensubsub2diff}); then we argue that
  the right-hand side of~\eqref{eq:coef_tay} satisfies the same
  recurrence equation. 

  One easily sees from (\ref{transgensubsub2diff}) that
\begin{equation}
\begin{split}
\label{icdiffeqtranssubsub2}
\bigl[s^{1}\bigr]f(s;u) &= \frac{u^{2}}{(u+1)^{5}},\\
\bigl[s^{2}\bigr]f(s;u) &= \frac{u^{2}(3u^{2}-4u+3)}{(u+1)^{7}}.
\end{split}
\end{equation}
Differentiating (\ref{transgensubsub2diff}) once with respect to $s$
leads to the differential equation
\begin{equation}
\label{diffeqtranssubsub2}
2s(1-s)(A-Bs)\frac{df(s;u)}{ds}+(4Bs^{2}-3Bs+As-2A)f(s;u)=0,
\end{equation}
where 
\[
A=(u+1)^{2} \quad \text{and} \quad  B=(u-1)^{2}.
\]

Substituting the expansion 
\[
f(s;u) = \sum_{k=1}^{\infty}p_{k}(u)s^{k}
\]
 into (\ref{diffeqtranssubsub2}) and equating the coefficients of
 $s^{k}$ gives
\begin{equation}
\label{eq:nrr}
2A(k-1)p_{k}(u)+\bigl(3A-B-2(A+B)k\bigr)p_{k-1}(u)+2Bk p_{k-2}(u)=0,
\end{equation}
with initial conditions given by
(\ref{icdiffeqtranssubsub2}). 

Inserting the right-hand side of~\eqref{eq:coef_tay} into the
recurrence relation~\eqref{eq:nrr} shows that the statement of the
lemma is equivalent to an identity among Jacobi polynomials, which can
be proved using the strategy outlined at the end of
Proposition~\ref{delaysubleadingtwoprop} and in Remark
\ref{linindepremark}. For this particular case it is more convenient
to express each Jacobi polynomial in the basis
$P^{(1,1)}_{j}(\tilde{u})$.
\end{proof}

It turns out that $f(s;u)$ is the generating function $\Delta
\mathcal{T}_2^{(2,0)}(u,s)$.  This result was first obtained using
periodic orbit theory by Berkolaiko and Kuipers\cite{BK11} --- more
precisely, they computed $\mathcal{T}_2^{(2,0)}(u,s)$, which is
related to $\Delta \mathcal{T}_2^{(2,0)}(u,s)$
by~\eqref{diffscorrespondence}. The RMT proof of this fact is a
particular case of the next proposition, where we obtain the same
generating function for Andreev billiards.

\begin{proposition}
\label{unitarysubsubbeta2prop}
Let $\beta=2$ and $\delta > -2$. The generating function for the
coefficients $\mathcal{T}_{k,2}^{(2,\delta)}(u)$ is
\begin{equation}
  \label{eq:nd2te}
  \mathcal{T}^{(2,\delta)}_{2}(u,s) =
  \frac{\delta^{2}us((u+1)^{2}-s(u-1)^{2})-
   4s^{2}u^{2}}{4(1-s)^{3/2}((u+1)^{2}-s(u-1)^{2})^{5/2}}.
\end{equation}
\end{proposition}
\begin{proof}
  Our techniques to compute the moments of the transmission
  eigenvalues are best suited to compute the differences of the
  moments.  Therefore, we shall study 
  \begin{equation}
\label{unitaryandreevgenfn}
\Delta \mathcal{T}^{(2,\delta)}_{2}(u,s) =
\frac{4su^{2}+\delta^{2}u\bigl(s(u-1)^{2}
-(u+1)^{2}\bigr)}{4\sqrt{1-s}\bigl((u+1)^{2}-s(u-1)^{2}\bigr)^{5/2}}
+\frac{u\delta^{2}}{4(u+1)^{3}}.
\end{equation}

  Expanding the coefficient~\eqref{coefficientexample} to order
  $O(n^{-1})$ and inserting it into~\eqref{transmissionexample} gives
\begin{equation}
\label{unitarysubsubbeta2expr}
\Delta \mathcal{T}_{k,2}^{(2,\delta)}(u) = \frac{1}{k}
\sum_{j=0}^{k}\binom{k}{j}\binom{k}{j-1}\frac{u^{2k-2j}}{(u+1)^{2k+3}}
F_{k,2}^{(2,\delta)}(u,j),
\end{equation}
where $F_{k,2}^{(2,\delta)}(u,j)$ is a polynomial quadratic in
$\delta$ and of $4$-th order in $u$; the definition of
$F_{k,2}^{(2,\delta)}(u,j)$ is reported in Appendix~\ref{ap:add_exs},
Eq.~\eqref{p2tksubsub}.  Thus, the right-hand side
of~\eqref{unitarysubsubbeta2expr} is a quadratic form in $\delta$ too.

In order to prove this proposition we need to show that
Eq.~\eqref{unitarysubsubbeta2expr} coincides with the $k$-th
Taylor coefficient of the right-hand side
of~\eqref{unitaryandreevgenfn} in a neighborhood of $s=0$ and for any
$u>1$.  The proof is systematic and we shall outline the main steps.  

Take, for example, 
\[
\bigl[\delta^0\bigr]\bigl[u^1\bigr]F_{k,2}^{(2,\delta)}(u,j) = 
-(k-j)(k-j+1)(2k-2j+1)/3
\]
and consider its contribution to the right-hand side of
Eq.~\eqref{unitarysubsubbeta2expr}: 
\begin{equation}
\label{eq:hyp_eq_b2_tc}
-\frac{2k(k-1)}{3}\sum_{j=0}^{k}\binom{k-1}{j}\binom{k-2}{j-1}
  \frac{u^{2k-2j}}{(u+1)^{2k+3}}
-\frac{k}{3}\sum_{j=0}^{k}\binom{k-1}{j}\binom{k-1}{j-1}
\frac{u^{2k-2j}}{(u+1)^{2k+3}}.
\end{equation}
The ratio of consecutive terms in these two sums are rational
functions of the summation indices.  Thus, they are special cases of
hypergeometric functions, which, in turn, can be expressed in terms of
Jacobi polynomials using the correspondence~\eqref{jacobihypergeom}.
Therefore, Eq.~\eqref{eq:hyp_eq_b2_tc} becomes
\begin{equation*}
 -\frac{u^{2}(u^{2}-1)^{k-2}}{(u+1)^{2k+3}}
  \left(2k(k-1)P^{(1,0)}_{k-2}(\tilde u)/3+kP^{(1,1)}_{k-2}(\tilde u)
   /3\right).
\end{equation*}
Repeating this procedure for each monomial in $u$ and $\delta$ in
$F_{k,2}^{(2,\delta)}(u,j)$, we see that
(\ref{unitarysubsubbeta2expr}) becomes a sum of different Jacobi
polynomials, leading to an explicit (though complicated) formula for
each coefficient of $\delta$. For convenience we report these formulae
in Appendix~\ref{ap:add_exs}, Eqs. (\ref{NTcoeffa0}),
(\ref{NTcoeffa1}) and (\ref{NTcoeffa2}). 

Up to now we have proved that $\Delta
\mathcal{T}^{(2,\delta)}_{k,2}(u)$ can be written as a particular
combination of Jacobi polynomials.  The rest of the proof proceeds as
follows: a) we shall express $\bigl[s^k\bigr]\Delta
\mathcal{T}^{(2,\delta)}_2(u,s)$ as a linear combination of Jacobi
polynomials too; b) by using the same procedure as in the proof of
Proposition~\ref{delaysubleadingtwoprop} --- points (1)--(4) --- it
follows that
\[
\bigl[s^k\bigr]\Delta
\mathcal{T}^{(2,\delta)}_2(u,s) - \Delta
\mathcal{T}^{(2,\delta)}_{k,2}(u) =0. 
\]
We shall only outline how to achieve a), as we have already discussed
how to prove b).

First note that $\Delta
\mathcal{T}^{(2,\delta)}_{2}(u,s)$ can be written in terms of $\Delta
\mathcal{T}^{(2,0)}_{2}(u,s)$:
\begin{equation*}
  \Delta \mathcal{T}^{(2,\delta)}_{2}(u,s)  = \Delta
  \mathcal{T}^{(2,0)}_{2}(u,s)\left(1+\frac{\delta^{2}(u-1)^{2}}{4u}-
    \frac{\delta^{2}(u+1)^{2}}{4us}\right) +\frac{u\delta^{2}}{4(u+1)^{3}}.
\end{equation*}
Thus, the Taylor coefficients of (\ref{unitaryandreevgenfn}) are 
\begin{equation}
\label{eq:tay_c_2del}
\begin{split}
\bigl[s^{k}\bigr]\Delta \mathcal{T}^{(2,\delta)}_{2}(u,s) & =
 \bigl[s^{k}\bigr]\Delta \mathcal{T}^{(2,0)}_{2}(u,s)+\delta^{2}
\biggl(\frac{(u-1)^{2}}{4u}\bigl[s^{k}\bigr]\Delta
  \mathcal{T}^{(2,0)}_{2}(u,s)\\
& \quad -\frac{(u+1)^{2}}{4u}\bigl[s^{k+1}\bigr]
\Delta \mathcal{T}^{(2,0)}_{2}(u,s)\biggr).
\end{split}
\end{equation}
Now, from Lemma \ref{unitarygenfntrans} and
Eq.~\eqref{unitaryandreevgenfn} we have $f(u;s)  = 
\Delta \mathcal{T}_2^{(2,0)}(u,s)$; therefore
\begin{equation*}
\frac{\bigl[s^{k}\bigr]6\Delta 
\mathcal{T}^{(2,0)}_{2}(u,s)(u+1)^{k+5}}%
{u^{2}k(k+1)(u-1)^{k-2}}=(1-u^{2})\biggl(k^{-1}P^{(1,1)}_{k-1}(\tilde{u})
-4P^{(0,0)}_{k-1}(\tilde{u})\biggr)
-4uP^{(1,1)}_{k-2}(\tilde{u}).
\end{equation*}
Inserting this expression into the right-hand side
of~\eqref{eq:tay_c_2del} gives $\bigl[s^k\bigr]\Delta
\mathcal{T}^{(2,\delta)}_2(u,s)$ in terms of Jacobi polynomials. 
\end{proof}

\section{Orthogonal and Symplectic Symmetries}
\label{se:orth_sym}

Our aim in this section is to compute the first two corrections to the
leading order terms (\ref{2leadjacobileading}) and
(\ref{wignerdelayleading}) when $\beta=1$ and $\beta=4$. The finite
$n$ formulae for the averages~\eqref{trans_mom}
and~\eqref{eq:mom_time_del} that we obtained\cite{MS11a} for
ensembles with both orthogonal and symplectic symmetries have a
similar structure. Thus, the techniques and calculations used to study
their asymptotic behaviour are almost identical for both symmetry
classes. Although we will state the results for $\beta=1$ and
$\beta=4$, most of the proofs will focus on $\beta=1$ ensembles, as
they contain the essential details of the arguments.

\subsection{From Finite-$n$ to Asymptotics --- Preliminaries} 
\label{sse:exact_results}
The finite-$n$ moments of the transmission eigenvalues and proper
delay times are
\begin{subequations}
\begin{align}
\label{orthogonaljacobiexample}
M_{\J}^{(1)}(k,n) & =M_{\J}^{(2)}(k,n-1) \\
& \quad - 2\sum_{j=1}^{\lfloor k/2 \rfloor}
\sum_{i=0}^{k-2j}\binom{k}{i}\binom{k}{i+2j}
S_{i,j}^{a/2,b/2}(k,(n-1)/2) \notag \\
&  \quad +I_{\J}(k,n) \notag
\intertext{and}
\label{orthogonalwishartexample}
M_{\Lag}^{(1)}(-k,n) & =M_{\Lag}^{(2)}(-k,n-1) \\
&\quad - 2^{1-k}\sum_{j=1}^{n/2-1}\sum_{i=0}^{n-2j}
\binom{k+j-1}{k-1}\binom{k+i+2j-1}{k-1}
S_{i,j}^{b/2}(-k,(n-1)/2)\notag \\
& \quad +I_{\Lag}(-k,n).  \notag
\end{align}
\end{subequations}
We give the explicit expressions of $S_{i,j}^{a/2,b/2}(k,n)$ and
$S_{i,j}^{b/2}(-k,n)$ in Eqs. (\ref{Sab}) and (\ref{Sb}); the
terms $I_{\J}(k,n)$ and $I_{\Lag}(-k,n)$ are defined in (\ref{eq:Iab})
and (\ref{Ib}) respectively. The quantities $M_{\J}^{(2)}(k,n-1)$ and
$M_{\Lag}^{(2)}(-k,n-1)$ are the moments for $\beta=2$, whose formulae
are given in Eqs.~\eqref{eq:expl_mom_j}
and~\eqref{wignerexample}. 

\begin{remark}
\label{re:beta_in}
  Consider the asymptotic expansion of $S_{i,j}^{a/2,b/2}(k,(n-1)/2)$
  and $I_{\J}(k,n)$.  Inserting the leading order
  of~\eqref{fullgammaratio} gives
\begin{subequations}
\begin{align}
\label{sijasymptex}
2S_{i,j}^{a/2,b/2}(k,(n-1)/2) &\sim 
\frac{u^{2k-2j-2i}}{(1+u)^{2k}}+O(n^{-1}), \quad  n \to \infty \\
\label{iknex}
I_{\J}(k,n) &\sim
\sum_{j=0}^{k}\binom{2k}{2j}\frac{u^{2j}}{(1+u)^{2k}}
+O(n^{-1}),\quad  n \to \infty.
\end{align}
\end{subequations}
The computation of the right-hand sides requires an analysis similar
to that one used to obtain~\eqref{leadingestimatetrans} for $\beta=2$.
It follows immediately that $S_{i,j}^{a/2,b/2}(k,n)$ and $I_{\J}(k,n)$
do not contribute at leading order. As a consequence the leading order
term in the expansion~\eqref{eq:asy_ex_tra} is independent of $\beta$
and $\delta$.  The same is true for the moments of the proper delay
times~\eqref{eq:asy_ex_td}. In other words, the contributions made by the
orthogonal and symplectic symmetries of the ensembles affect only the
higher order terms.  This is a general property of the density of the
eigenvalues in random matrix ensembles.
\end{remark}

Higher order coefficients in the right-hand sides of
~\eqref{eq:asy_ex_tra} and \eqref{eq:asy_ex_td} are computed by
determining the contributions of higher order terms in the expansions
of $S_{i,j}^{a/2,b/2}(k,n)$, $S_{i,j}^{b/2}(k,n)$, $I_{\J}(k,n)$ and
$I_{\Lag}(-k,n)$.  Thus, for example, in order to compute a given
contribution in the expansion of~\eqref{orthogonaljacobiexample}, we
shall separately study the series
\begin{subequations}
  \begin{align}
    \label{uniexpansions}
    \frac{1}{n}M_{\mathrm{J}_{a,b}}^{(2)}(k,n-1) & \sim 
  \sum_{p=0}^{\infty}\mathcal{U}^{\J}_{k,p}(u)n^{-p}, \\
  \label{jacsympexpansion}
\sum_{j=1}^{\lfloor k/2 \rfloor}\sum_{i=0}^{k-2j}\binom{k}{i}
\binom{k}{i+2j}2S_{i,j}^{a/2,b/2}(k,(n-1)/2) 
& \sim \sum_{p=0}^{\infty}\mathcal{S}^{\J}_{k,p}(u)n^{-p} \\
\intertext{and}
\label{iknexexpan}
I_{\J}(k,n)&  \sim \sum_{p=0}^{\infty}\mathcal{I}^{\J}_{k,p}(u)n^{-p}
  \end{align}
\end{subequations}
as $n \to \infty$.  It follows that the coefficients in 
the expansion~\eqref{eq:asy_ex_tra} can be represented as 
\begin{equation}
\label{transasymptrule}
\mathcal{T}_{k,p}^{(1,\delta)}(u) = \mathcal{U}_{k,p}^{(1,\delta)}(u) 
- \mathcal{S}^{(1,\delta)}_{k,p}(u)+\mathcal{I}^{(1,\delta)}_{k,p}(u).
\end{equation}
For convenience, we have expressed the parameters $a$ and $b$ of the
Jacobi ensemble in terms of $\beta$ and $\delta$ (see
Eq.~\eqref{eq:jacobi_parameters}); thus, in~\eqref{transasymptrule} we
have replaced the superscript $\J$ with $(\beta,\delta)$. We shall
adopt this notation in the rest of the paper.  Similarly, for the
proper delay times we shall use the notation $(\beta)$ instead of
$\Lag$ (see Eq.~\eqref{eq:lag_par}). The contributions to other
quantities can be broken down in the same way.  Therefore, we shall
write
\begin{subequations}
\begin{align}
\label{transasymptruled}
\Delta \mathcal{T}_{k,p}^{(1,\delta)}(u) & = \Delta 
\mathcal{U}_{k,p}^{(1,\delta)}(u) - 
\Delta \mathcal{S}^{(1,\delta)}_{k,p}(u)+\Delta \mathcal{I}^{(1,\delta)}_{k,p}(u)\\
\label{delayasymptruled}
\mathcal{D}_{k,p}^{(1)}(w) & = \mathcal{U}^{(1)}_{-k,p}(w) - 
\mathcal{S}^{(1)}_{-k,p}(w)
+\mathcal{I}^{(1)}_{-k,p}(w). 
\end{align}
\end{subequations}

Without loss of generality, we shall refer to
$\mathcal{U}_{k,p}^{(1,\delta)}(u)$, $\Delta \mathcal{U}_{k,p}^{(1,\delta)}(u)$ and
$\mathcal{U}^{(1)}_{-k,p}(w)$ as \emph{unitary contributions;} similarly, we shall call
$\mathcal{S}_{k,p}^{(1,\delta)}(u)$, $\Delta \mathcal{S}_{k,p}^{(1,\delta)}(u)$ and
$\mathcal{S}^{(1)}_{-k,p}(w)$ \emph{symplectic contributions.}

\subsection{Next to Leading Order --- Negative Moments in the
  Laguerre Ensembles}

\begin{proposition}
\label{subleadingwignerdelayrmt}
Let $\beta \in \{1,2,4\}$ , the next to leading order term in the
expansion~\eqref{eq:asy_ex_tra} is 
\begin{equation}
\label{subleadingwignerdelayeqn}
\mathcal{D}_{k,1}^{(\beta)}(w) = \frac{1}{2(w-1)^{2k}}
\left(\frac{2}{\beta}-1\right)
\sum_{j=0}^{k}\left(\binom{2k}{2j}
-\binom{k}{j}^{2}\right)w^{j}.
\end{equation}
\end{proposition}

\begin{proof}
  When $\beta=2$ the right-hand side
  of~\eqref{subleadingwignerdelayeqn} is zero, which is consistent
  with Proposition~\ref{unitarytermsvanishing}; we shall only discuss
  the proof for $\beta=1$, as it is identical to that for $\beta=4$.

  Take the next to leading order term as $n\to \infty$ of
  formulae~\eqref{Sb} and~\eqref{Ib}. Then, by proceeding as outlined
  in Sec.~\ref{sse:exact_results} we arrive at
\begin{subequations}
\begin{align}
\label{unideriv1}
\mathcal{U}^{(1)}_{-k,1}(w) & = -\sum_{j=0}^{\infty}
\binom{k+j-1}{k-1}\binom{k+j}{k}w^{-k-j},\\
\label{incderiv1}
\mathcal{I}^{(1)}_{-k,1}(w) & = \sum_{j=0}^{\infty}
\binom{2k+2j-1}{2k-1}w^{-k-j}. 
\end{align}
\end{subequations}

To obtain the next to leading order symplectic contribution we need to
work a little bit more.  We obtain
\begin{align}
\label{lastlinederiv1}
  \mathcal{S}^{(1)}_{-k,1}(w) &= \sum_{j=1}^{\infty}
  \sum_{i=0}^{\infty}\binom{k+i-1}{k-1}\binom{k+2j+i-1}{ k-1}w^{-k-i-j}  \\
\label{midlinederiv1}  
&= \sum_{j=0}^{\infty}\sum_{i=0}^{j}\binom{k+j-i-1}{k-1}
  \binom{k+j+i+1}{k-1}w^{-k-j-1}\tag{93'} \\
  &=\frac{1}{2}\left(\sum_{j=0}^{\infty}\binom{2k+2j+1}{2k-1}w^{-k-j-1}
  -\sum_{j=0}^{\infty}\binom{k+j}{k-1}^{2}w^{-k-j-1}\right).\notag 
\end{align}

The inner sum in line~\eqref{midlinederiv1} was evaluated using
Chu-Vandermonde's summation (see Example~\ref{ex:chu_sum1},
Appendix~\ref{sse:chu_vandermonde}).  Combining Eqs.
(\ref{unideriv1}), (\ref{incderiv1}) and (\ref{lastlinederiv1})
into~(\ref{delayasymptruled}) yields
\begin{equation}
\label{eq:temp_result}
\mathcal{D}_{k,1}^{(1)}(w) = \frac{1}{2}\left(\sum_{j=0}^{\infty}
\binom{2k+2j-1}{2k-1}w^{-k-j}-(w-1)
\sum_{j=0}^{\infty}\binom{k+j}{k}^{2}w^{-k-j-1}\right).
\end{equation}

The identity
\begin{equation}
\label{subleadingsimplesum}
(w-1)^{-2k}\sum_{j=0}^{k}\binom{2k}{2j}w^{j} 
= \sum_{j=0}^{\infty}\binom{2k+2j-1}{2k-1}w^{-k-j},
\end{equation}
which is a simple consequence of the binomial theorem, takes care of
the first series in~\eqref{eq:temp_result}; the second sum is a
hypergeometric function:
\begin{equation*}
\begin{split}
\sum_{j=0}^{\infty}\binom{k+j}{k}^{2}w^{-k-j-1} & = 
w^{-k-1}{}_2F_1(k+1,k+1,1,w^{-1})\\
&=(w-1)^{-k-1}P^{(0,0)}_{k}(\tilde{w}).
\end{split}
\end{equation*}
The explicit representation of the Jacobi polynomials~\eqref{jacdef} 
combined with~\eqref{subleadingsimplesum} completes the proof.
\end{proof}

\begin{corollary}
\label{co:gf_td_1}
The generating function of the next to leading order corrections of
the moments~\eqref{subleadingwignerdelayeqn} is
\begin{equation}
\label{momslaggenp1prop}
\begin{split}
\mathcal D^{(\beta)}_{1}(w,s) & =\frac12 \left( \frac{2}{\beta} -1\right)\left(
\frac{(w-1)^{2}-(w+1)s}{s^{2}-2s(w+1)+(w-1)^{2}}\right.\\
 & \quad \left.- \frac{(w-1)\sqrt{s^{2}-2s(w+1)+(w-1)^{2}}}{s^{2}-2s(w+1)+(w-1)^{2}}\right).
\end{split}
\end{equation}
\end{corollary}
\begin{proof}
We now show that the moments~\eqref{subleadingwignerdelayeqn} are the
Taylor coefficients of the function~\eqref{momslaggenp1prop}. 

Consider  the identity
\begin{equation}
  \sum_{j=0}^{k}\binom{2k}{2j}w^{j} = 
\frac{1}{2}\left((1-\sqrt{w})^{2k}+(1+\sqrt{w})^{2k}\right).
\end{equation}
For $s$ sufficiently small we can write
\begin{equation}
\label{2k2jgenfn}
\begin{split}
\frac{1}{2}\sum_{k=1}^{\infty}\sum_{j=0}^{k}\binom{2k}{2j}w^{j}s^{k} 
&= \frac{1}{4}\left(\frac{s(1-\sqrt{w})^{2}}%
{1-s(1-\sqrt{w})^{2}}+\frac{s(1+\sqrt{w})^{2}}{1-s(1+\sqrt{w})^{2}}\right)\\
&=\frac{1-s(w+1)}{2(s^{2}(w-1)^{2}-2s(w+1)+1)}-\frac{1}{2}. 
\end{split}
\end{equation}

A direct inspection of the definition of the Jacobi
polynomials~\eqref{jacdef} gives the generating function
\begin{equation}
\label{eq:fs_gf}
\frac{1}{2}\sum_{k=1}^{\infty}\sum_{j=0}^{k}\binom{k}{j}^{2}w^{j}s^{k} 
= \frac{1}{2}\sum_{k=1}^{\infty}(w-1)^{k}P_{k}^{(0,0)}(\tilde{w})s^{k}.
\end{equation}
Recall that $\tilde{w}=(w+1)/(w-1)$.  Furthermore, the generating
function of the Jacobi polynomials $P_k^{(\alpha,\beta)}(w)$ when
$\alpha = \beta=0$ is (see Ref.~\citenum{AS72}, Eq. (22.9.1))
\begin{equation}
\label{eq:jp_gf}
\sum_{k=1}^{\infty}P_{k}^{(0,0)}(w)s^{k} = \frac{1}{\sqrt{s^{2}-2ws+1}}-1.
\end{equation}
Combining Eqs.~\eqref{eq:fs_gf} and~\eqref{eq:jp_gf} leads to
\begin{equation}
\label{genfnsquarebinoms}
\frac{1}{2}\sum_{k=1}^{\infty}\sum_{j=0}^{k}\binom{k}{j}^{2}w^{j}s^{k}
=\frac{1}{2\sqrt{s^{2}(w-1)^{2}-2(w+1)s+1}}-\frac{1}{2}.
\end{equation}

Finally, to complete the proof we introduce the scaling $s \to
(w-1)^{-2}s$ and subtract \eqref{genfnsquarebinoms} from
(\ref{2k2jgenfn}).
\end{proof}

\begin{remark}
When $w=2$ and $\beta=1$ Eq.~\eqref{momslaggenp1prop} reduces to
\begin{equation*}
  \mathcal D^{(1)}_1(2,s) = \frac{1-3s-\sqrt{1-6s+s^{2}}}{2(1-6s+s^{2})},
\end{equation*}
which was computed from periodic orbit theory.\cite{BK11}

\end{remark}

\subsection{Next to Leading Order --- Transmission Eigenvalues}
The following proposition extends
formula~(\ref{transfirstunitarycorr}) to the cases $\beta=1$ and
$\beta=4$. In Ref.~\citenum{DBB10}, Eq.~ (9), they obtain the
next to leading order correction of the eigenvalue density, though no
proof was given. For the Wigner-Dyson symmetries, $\delta=0$; the
corresponding result first appeared in Ref.~\citenum{BB96} and was
derived via integration over matrix elements combined with a
perturbative approach based on diagrammatic methods. We will present a
derivation using a completely alternative method.

\begin{proposition}
\label{pr:ntl_teig}
For $\beta \in \{1,2,4\}$ and $\delta > -2$ one
has
\begin{equation}
\label{jacobisubleadingchaos}	
\begin{split}
\Delta \mathcal{T}^{(\beta,\delta)}_{k,1}(u) &= 
\left(\frac{2}{\beta}-1\right)\frac{u}{(u+1)^{2}}
\left(\frac{u-1}{u+1}\right)^{2k}\\
&\quad +\frac{\delta}{\beta}\frac{1}{(u+1)^{2k+2}}
\left(\sum_{j=0}^{k}\binom{k}{j}^{2}u^{2j+1}-\sum_{j=0}^{k}\binom{k}{j}
\binom{k}{j+1}u^{2j+2}\right). 
\end{split}
\end{equation}
\end{proposition}
\begin{proof}
  As previously we shall discuss only the symmetry $\beta=1$. Let us
  set $a=\delta+1$, $b=(u-1)n$ in formula
  (\ref{orthogonaljacobiexample}). 

Following the general strategy
  outlined in Sec.~\ref{sse:exact_results}, we see that the unitary
  contribution in (\ref{uniexpansions}) is 
\begin{equation*}
  \Delta \mathcal{U}^{(1,\delta)}_{k,1}(u) =
  \frac{1}{k}\sum_{j=0}^{k}
\binom{k}{j}\binom{k}{j-1}
\frac{u^{2k-2j+1}}{(u+1)^{2k+2}}F^{(1,\delta)}_{k,1}(u,j),
\end{equation*}
where 
\[
F^{(1,\delta)}_{k,1}(u,j) = (1+k-j)(k+\delta-j-1)+uk(1-\delta)
+u^{2}j(\delta-j).
\]
Simplifying, we find
\begin{equation}
\label{firstcorunitrans1}
\Delta \mathcal{U}^{(1,\delta)}_{k,1}(u) =
\frac{(\delta-1)}{(u+1)^{2k+2}}
\left(\sum_{j=0}^{k}\binom{k}{j}^{2}u^{2j+1}
-\sum_{j=0}^{k-1}\binom{k}{j}\binom{k}{j+1}u^{2j+2}\right),
\end{equation}
which is obtained using the same method which led from
(\ref{easysimptra}) to the equation (\ref{floorceilingpoly}), provided
that the we substitute $a=\delta+1$ and $n \to n-1$ in
(\ref{coefficientexample}). 

The symplectic contribution at order $p=1$ is obtained by replacing
(\ref{sijasymptex}) into (\ref{jacsympexpansion}).  Using the
Chu-Vandermonde summation of Lemma~\ref{chupropjacobi}, we obtain
\begin{equation}
\label{subleadingdoublesum}
\begin{split}
\mathcal{S}^{(1,\delta)}_{k,1}(u) &=\sum_{j=1}^{k-1}
\sum_{i=1}^{j}\binom{k}{j-i}\binom{k}{j+i}\frac{u^{2j}}{(1+u)^{2k}}\\
& =\frac{1}{2}\sum_{j=1}^{k-1}\left(\binom{2k}{2j}
-\binom{k}{j}^{2}\right)\frac{u^{2j}}{(1+u)^{2k}}.
\end{split}
\end{equation}

Finally, we must examine the term $I_{\J}(k,n)$ in
(\ref{orthogonaljacobiexample}).  Straightforward application of
(\ref{iknex})
gives
\begin{equation}
\label{eq:DI_Jac_1}
\Delta \mathcal{I}^{(1,\delta)}_{k,1}(u) =
2u(u-1)^{2k}(u+1)^{-2k-2}.
\end{equation}
Inserting (\ref{subleadingdoublesum}), (\ref{firstcorunitrans1}) and
\eqref{eq:DI_Jac_1} into (\ref{transasymptruled}) shows that equation
(\ref{jacobisubleadingchaos}) is equivalent to proving the
combinatorial identity
\begin{equation}
\begin{split}
\label{finaltransfirstcorr}
u(u-1)^{2k} & =
(u+1)^{2k+2}\Delta\mathcal{S}^{(1,\delta)}_{k,1}(u)\\
& \quad  +\left(\sum_{j=0}^{k}\binom{k}{j}^{2}u^{2j+1}
  -\sum_{j=0}^{k}\binom{k}{j}\binom{k}{j+1}u^{2j+2}\right)
\end{split}
\end{equation}
where $\Delta\mathcal{S}^{(\beta,\delta)}_{k,1}(u) := 
\mathcal{S}^{(1,\delta)}_{k,1}(u)-\mathcal{S}^{(1,\delta)}_{k+1,1}(u)$ 
is given in~(\ref{subleadingdoublesum}). 

Equation~\eqref{finaltransfirstcorr} can be proved simply by comparing
the coefficients of the monomials $u^{2j}$ ($u^{2j+1}$) in the
polynomials on both sides of the equation. It turns out that the
right-hand side can be reduced to the left-hand side by repeated
applications of Pascal's rule
\[
\binom{k}{j} + \binom{k}{j+1} = \binom{k+1}{j+1}.
\]
\end{proof}
Using the identity~\eqref{eq:vivocombidentity}, we can write the
statement of Proposition \ref{pr:ntl_teig} in terms of the leading
order~\eqref{marcelasymptresult}. Thus we have,
\begin{corollary}
\label{co:fc_te_1}
  The first correction for the differences of moments can be written
  in terms of formula~\eqref{marcelasymptresult} as
\begin{equation}
\label{eq:fc_te}
\Delta \mathcal{T}^{(\beta,\delta)}_{k,1}(u)=\frac{u}{(u+1)^{2}}
\left(\frac{u-1}{u+1}\right)^{2k}\left(\frac{2}{\beta}-1
+\frac{\delta(u-1)}{\beta u}\mathcal{T}^{(\beta,\delta)}_{k+1,0}(-u)\right)
\end{equation}
\end{corollary}
  %The proof follows the same steps of the derivation of the next to
  %leading order contribution for
  %$\beta=2$~\eqref{transfirstunitarycorr} in
  %Proposition~\ref{unitarytermsvanishing}.
  %Eventually, one arrives at 
%\[
%\frac{\delta}{\beta}\frac{1}{(u+1)^{2k+2}}
%\left(\sum_{j=0}^{k}\binom{k}{j}^{2}u^{2j+1}-\sum_{j=0}^{k}\binom{k}{j}
%\binom{k}{j+1}u^{2j+2}\right), 
%\]
%which is the analogue of~\eqref{floorceilingpoly}.  Finally,
%Eq.~\eqref{eq:fc_te} is a straightforward consequence of
%the identity~\eqref{eq:vivocombidentity}.

From Corollary~\ref{co:fc_te_1} we can compute the generating function
of~\eqref{eq:fc_te}.
\begin{corollary}
  \label{co:ntogfte}
  The next to leading order generating function of the moments of the
  transmission eigenvalues is
  \begin{equation}
  \label{eq:ntogfte2}
  \begin{split}
\mathcal T^{(\beta,\delta)}_{1}(u,s) &= \left(\frac{2}{\beta}-1\right)
  \left(\frac{us}{(s-1)\bigl((u+1)^{2}-s(u-1)^{2}\bigr)}\right)\\
    &\quad +\frac{\delta}{2\beta}\left(\frac{(u+1)}{\sqrt{1-s}
   \sqrt{(u+1)^{2}-s(u-1)^{2}}}+\frac{1}{s-1}\right).
\end{split}
\end{equation}
\end{corollary}
\begin{proof}
 The generating function of the factor
 $\mathcal{T}_{k+1,0}^{(\beta,\delta)}(-u)$ in
 Eq.~\eqref{eq:fc_te} is easily obtained
 from~\eqref{eq:lead_gen_fun_te}. Then,  the substitution
 \begin{equation*}
    s \to s \left(\frac{u-1}{u+1}\right)^2
 \end{equation*}
and some algebra lead to
  \begin{equation}
    \label{eq:ntogfte}
    \begin{split}
     \Delta \mathcal{T}^{(\beta,\delta)}_{1}(u,s) &=
     \left(\frac{2}{\beta}-1\right)
          \left(\frac{u}{(u+1)^{2}-s(u-1)^{2}}-\frac{u}{(u+1)^{2}}\right)\\
             &\quad -\frac{\delta}{2\beta}\left(\frac{(u+1)
             \sqrt{1-s}}{s\sqrt{(u+1)^{2}-s(u-1)^{2}}}+\frac{2u}{(u+1)^{2}}
           -\frac{1}{s}\right).
    \end{split}
  \end{equation}
Finally, formula~\eqref{eq:ntogfte2} follows from~\eqref{diffscorrespondence}.
\end{proof}

\subsection{Second Corrections}
\label{se:scbeta1}

As we go higher in the order of the expansions~\eqref{eq:asy_ex_tra}
and \eqref{eq:asy_ex_td} the formulae for the coefficients become very
cumbersome.  Thus, the results are better expressed only in terms of
generating functions.  For the sake of simplicity, we shall only
discuss the symmetry classes with $\beta=1$.

\subsubsection{Negative Moments in the Laguerre Orthogonal Ensemble}
\begin{proposition}
\label{delaysubsub1prop}
The generating function of the second corrections to the moments of
the proper delay times is 
\begin{equation}
\label{delaysubsub1genl1}
\begin{split}
\mathcal D^{(1)}_2(w,s) &= s\frac{(w+1)s^{2}-(2w^{2}-3w+2)s
+(w-1)^{2}(w+1)}{(s^{2}-2(w+1)s+(w-1)^{2})^{5/2}} \\
&\quad +s\frac{(w-1)s+1-w^{2}}{(s^{2}-2(w+1)s+(w-1)^{2})^{2}} 
%\label{delaysubsub1genl2}
\end{split}
\end{equation}
\end{proposition}
\begin{proof}
 The proof is split into two parts: first we compute the third
 coefficients in the asymptotic expansion of
 formula~\eqref{orthogonalwishartexample}; then we show that they
 coincide with the Taylor coefficients of the right-hand side of
 Eq.~\eqref{delaysubsub1genl1}.
   
 We start from the coefficients of the unitary contribution
 $\mathcal{U}_{k,2}^{(1)}(w)$. Inserting $n \to n-1$ and then
 $b=1+(w-1)n$ in (\ref{wignerexample}) we find
\begin{equation}
\label{eq:usc} 
\begin{split}
\mathcal{U}_{k,2}^{(1)}(w) & = \frac{w(k+1)}{(w-1)^{k+3}}
\biggl((w-1)\bigl(P_{k-1}^{(1,1)}(\tilde{w})-(k/2)
P_{k}^{(0,1)}(\tilde{w})\bigr) \\
& \quad  -\frac{2+k}{12}P_{k-1}^{(1,2)}(\tilde{w})
+\frac{7(k+2)}{12}P_{k-1}^{(2,1)}(\tilde{w})\biggr).
\end{split}
\end{equation}
The derivation of this formula is almost identical to the discussion
of Eqs.~\eqref{polyexample2k_1} and (\ref{polyexampled2k}) in the
proof of Proposition~\ref{delaysubleadingtwoprop}.

The term involving the asymptotics of $I_{\mathrm{L}_{b}}(-k,n)$ is
the easiest to compute; it is an application of the
expansion~\eqref{subleadingpochhammer} to~\eqref{Ib}. Its
contribution is
\begin{equation}
\label{eq:Isc}
\mathcal{I}_{k,2}^{(1)}(w) 
=-\frac{k}{2(w-1)^{2k+1}}\sum_{j=0}^{k}\binom{2k+2}{2j+1}w^{j}.
\end{equation}

The coefficient of the symplectic contribution leads to the double sum
\begin{equation}
\label{sympdelaysubsub1polydef}
\mathcal{S}_{k,2}^{(1)}(w)
= \sum_{j=1}^{\infty}\sum_{i=0}^{\infty}\binom{k+i-1}{k-1}
\binom{k+2j+i-1}{k-1}w^{-k-i-j-1}G^{(1)}_{k,2}(w,i,j),
\end{equation}
where 
\[
G^{(1)}_{k,2}(w,i,j) = (k+i+2j-1)(k+i+2j)/2-j(j-1) 
-w\bigl(j^{2}+i(4j+i+1)/2\bigr). 
\]
Now consider, for example, the contribution from the coefficient of
$w$ in $G^{(1)}_{k,2}(w,i,j)$. We have
\begin{multline*}
\sum_{j=1}^{\infty}\sum_{i=0}^{\infty}\binom{k+i-1}{k-1}
\binom{k+2j+i-1}{k-1}w^{-k-i-j}\bigl(j^{2}+i(4j+i+1)/2\bigr)\\
=\sum_{j=1}^{\infty}\sum_{i=1}^{j}
\binom{k+j-i-1}{k-1}\binom{k+j+i-1}{k-1}w^{-k-j}\bigl(i^{2}+(j-i)
(4i+(j-i)+1)/2\bigr).
\end{multline*}
The finite sum in this equation can be computed exactly in terms of
binomial coefficients using Chu-Vandermonde's summation techniques
discussed in Appendix \ref{sse:chu_vandermonde}. The remaining
infinite series is a combination of hypergeometric sums which may be
evaluated in terms of Jacobi polynomials using the identity
(\ref{wignerhypergeom}). Repeating this procedure and inserting the
contributions~\eqref{eq:usc}, \eqref{eq:Isc} and
$\mathcal{S}_{k,2}^{(1)}(w)$ into (\ref{delayasymptruled}) we arrive at 
\begin{equation}
\label{delaysubsubleading1final}
\begin{split}
  2(w-1)^{k+2}\mathcal{D}_{k,2}^{(1)}(w) &= -k(w-1)
  P^{(1,0)}_{k-1}(\tilde{w})-k(k+1)\bigl(P^{(1,0)}_{k}(\tilde{w})
  +wP^{(0,1)}_{k}(\tilde{w})\bigr)  \\
  &\quad +(w/6)(w-1)^{-1}(k+1)(k+2)\bigl(13P^{(2,1)}_{k-1}(\tilde{w})
  -P^{(1,2)}_{k-1}(\tilde{w})\bigr)  \\
  &\quad -\frac{k}{2(w-1)^{k-1}}\sum_{j=0}^{k}\binom{2k+2}{2j+1}w^{j}\\
  & \quad -kwP^{(0,2)}_{k-1}(\tilde{w}) -\bigl(1-3w(k+1)\bigr)
   P^{(1,1)}_{k-1}(\tilde{w}).
\end{split}
\end{equation}

We need to show that the Taylor coefficients of the right-hand side of
Eq.~\eqref{delaysubsub1genl1} are given by
$\mathcal{D}_{k,2}^{(1)}(w)$ in~\eqref{delaysubsubleading1final}.

First notice that
\begin{equation}
\label{eq:alth}
  -\sum_{k=0}^{\infty}\frac{k}{4(w-1)^{2k+1}}
\sum_{j=0}^{k}\binom{2k+2}{2j+1}w^{j}s^{k} 
= s\frac{(w-1)s+1-w^{2}}{(s^{2}-2(w+1)s+(w-1)^{2})^{2}},
\end{equation}
which follows from elementary manipulations involving the binomial
theorem and geometric series.  Then, a straightforward application of
Corollary~\ref{lemma1wiggen} gives
\begin{multline}
\label{wiggen2ndcorrl3}
\bigl[s^{k}\bigr]\frac{(w+1)s^{3}-(2w^{2}-3w+2)s^{2}+(w-1)^{2}(w+1)s}%
{(s^{2}-2(w+1)s+(w-1)^{2})^{5/2}}\\
=\frac{(w+1)}{w}\mathcal{D}_{k-1,2}^{(2)}(w) 
-\frac{(2w^{2}-3w+2)}{w}\mathcal{D}_{k,2}^{(2)}(w)
+\frac{(w-1)^{2}(w+1)}{w}\mathcal{D}_{k+1,2}^{(2)}(w).
\end{multline}
Insert the explicit formula~\eqref{delaysubleadingtwo} for
$\mathcal{D}_{k,2}^{(2)}(w)$ into the right-hand side
of Eq.~\eqref{wiggen2ndcorrl3} and apply the three term recurrence
realation~\eqref{threetermrec} to remove $P^{(2,2)}_{k-3}(\tilde{w})$.
We obtain
\begin{equation*}
\frac{(w+1)P^{(2,2)}_{k-1}(\tilde{w})(k+3)}{4(w-1)^{k+2}}
+\frac{(8w+7wk-3w^{2}-3)(k+1)P^{(2,2)}_{k-2}(\tilde{w})}{12(w-1)^{k+3}}.
\end{equation*}
Combining this expression with (\ref{eq:alth}) leads to
\begin{equation}
\label{eq:tcgf}
\begin{split}
\bigl[s^{k}\bigr] \mathcal D_{2}^{(1)}(w,s) &=
\frac{(w+1)P^{(2,2)}_{k-1}(\tilde{w})(k+3)}%
{4(w-1)^{k+2}}-\frac{k}{4(w-1)^{2k+1}}\sum_{j=0}^{k}\binom{2k+2}{2j+1}w^{j}
\\ 
& \quad +\frac{(8w+7wk-3w^{2}-3)(k+1)P^{(2,2)}_{k-2}(\tilde{w})}%
{12(w-1)^{k+3}}.
\end{split}
\end{equation}

The demonstration that coefficients~\eqref{delaysubsubleading1final}
are equivalent to the right-hand side of Eq.~\eqref{eq:tcgf}
involves cumbersome algebraic manipulations, but it is otherwise
straightforward.  It can be carried out systematically following the
steps (1)--(4) of the procedure outlined at the end of the proof of
Proposition~\ref{delaysubleadingtwoprop} and further discussed in
Remark~\ref{linindepremark}.
\end{proof}

\subsubsection{Moments in the Jacobi Orthogonal Ensemble}

As for the second correction of the moments of the proper delay times
the results for the transmission eigenvalues are better expressed only
in terms of generating functions.    For simplicity, we first discuss
the moments of the transmission  for the classical Dyson's ensemble,
\textit{i.e.}  when $\delta=0$.

\begin{proposition}
\label{transsubsub1prop}
The second correction of the moments of the transmission eigenvalues
for $\beta=1$ and $\delta=0$ is
\begin{equation}
\label{transgensubsub1diff}
\mathcal T^{(1,0)}_{2}(u,s) = -\frac{us\bigl(s^{2}(u-1)^{2}+3us
-(u+1)^{2}\bigr)}{\bigl((u+1)^{2}-s(u-1)^{2}\bigr)^{5/2}(1-s)^{3/2}}.
\end{equation}
\end{proposition}
\begin{proof}
  The proof follows the usual strategy: firstly, we compute the
  coefficient at $p=2$ of the expansion~\eqref{eq:asy_ex_tra}; then we
  show that they coincide with the Taylor coefficients
  of~\eqref{transgensubsub1diff}.  Furthermore, rather 
  than computing~\eqref{transgensubsub1diff} directly we will look at
\begin{equation}
\label{eq:gen_diff}
  \Delta \mathcal T^{(1,0)}_{2}(u,s) =
  \frac{u\bigl(s^{2}(u-1)^{2}+3us-(u+1)^{2}\bigr)}%
  {\bigl((u+1)^{2}-s(u-1)^{2}\bigr)^{5/2}\sqrt{1-s}}
  +\frac{u}{(u+1)^{3}}.
\end{equation}

According to~\eqref{transasymptruled}, to begin with we split $\Delta
\mathcal{T}^{(1,0)}_{k,2}(u)$ into the contributions of $\Delta
\mathcal{U}^{(1,0)}_{k,2}(u)$, $\Delta \mathcal{S}^{(1,0)}_{k,2}(u)$
and $\Delta \mathcal{I}^{(1,0)}_{k,2}(u)$. First we note that
straightforward calculation analogous to the proof that
$\Delta \mathcal{T}^{(2,0)}_{k,1} (u) = 0$ in
Proposition~\ref{unitarytermsvanishing} shows that $\Delta
\mathcal{I}^{(1,0)}_{k,2}(u)=0$; therefore, we need not worry about
this term.

The contribution of the symplectic term involving the double summation
in (\ref{jacsympexpansion}) can be written as
\begin{equation}
\label{symptranssubsub1polydef}
\mathcal{S}_{k,2}^{(1,0)}(u) = \frac{1}{(u+1)^{2k+1}}
\sum_{j=1}^{\lfloor k/2 \rfloor}
\sum_{i=0}^{k-2j}\binom{k}{i}\binom{k}{i+2j}u^{2k-2i-2j-1}
F^{(1)}_{k,2}(u,i,j),
\end{equation}
where 
\[
F^{(1)}_{k,2}(u,i,j) = (j+i^{2}+4ij+2j^{2}-2ki+k^{2}-4kj) 
+2ukj+u^{2}(j-i^{2}-2j^{2}-4ij) .
\]
For example, taking the coefficient of $u^{2}$ in
$F^{(1)}_{k,2}(u,i,j)$ leads to 
\begin{multline}
\label{symptransexpoly}
  \frac{1}{(u+1)^{2k+1}}\sum_{j=1}^{\lfloor k/2 \rfloor}
  \sum_{i=0}^{k-2j}\binom{k}{i}\binom{k}{i+2j}u^{2k-2i-2j+1}(j-i^{2}-2j^{2}-4ij)
  \\
  =\frac{1}{(u+1)^{2k+1}}\sum_{j=1}^{k-1}
\sum_{i=1}^{j}\binom{k}{j+i}\binom{k}{j-i}u^{2k-2j+1}
\bigl(i-(j-i)^{2}-2i^{2}-4(j-i)i\bigr).
\end{multline}
The inner sums in this equation are of the form 
\begin{equation}
\sum_{i=1}^{j}\binom{k}{j+i}\binom{k}{j-i}i^{p}, \quad p=0,1,2
\end{equation}
They are Chu-Vandermonde's summations discussed in
Appendix~\ref{sse:chu_vandermonde} and can be evaluated explicitly in
terms of binomial coefficients.  This allows us to
write~\eqref{symptransexpoly} as a sum of Jacobi polynomials.
Finally, including the contributions of the linear and constant terms
in $u$ in $F^{(1)}_{k,2}(u,i,j)$ we arrive at
\begin{equation}
\label{eq:sc2te}
\begin{split}
\frac{(u+1)^{k+3}}{(u-1)^{k-2}u}\mathcal{S}_{k,2}^{(1,0)}(u) 
&= k(k-1)(1-u^{2})P^{(0,1)}_{k-2}(\tilde{u})+(k/2)(1-2k)P^{(1,1)}_{k-2}(\tilde{u})  \\
& \quad +(k/2)u^{2}P^{(1,1)}_{k-2}(\tilde{u})+k^{2}uP^{(1,1)}_{k-2}(\tilde{u}),
\end{split}
\end{equation}
where $\tilde{u}$ is defined in~\eqref{eq:u_tilde}

Now, $(u+1)^{2k +3}\Delta \mathcal{T}^{(1,0)}_{k,2}(u)$ is a polynomial
in $u$ of degree $2k+1$; for convenience we separate it  into the sum
of two polynomials containing only monomials of even and odd degree
respectively. We write
\begin{equation}
\label{eq:totmom2}
(u+1)^{2k+3}\Delta \mathcal{T}^{(1,0)}_{k,2}(u) = 
\Delta \mathcal{T}^{(1,0),\mathrm o}_{k,2}(u)+\Delta
\mathcal{T}^{(1,0),\mathrm e}_{k,2}(u)
\end{equation}
Let us focus first on the term containing only odd monomials. We have 
\begin{equation}
\label{oddmonomialsrule}
\Delta \mathcal{T}^{(1,0),\mathrm o}_{k,2} (u) = 
\Delta \mathcal{U}^{(1,0),\mathrm o}_{k,2}(u)-
\Delta \mathcal{S}^{(1,0),\mathrm o}_{k,2}(u).
\end{equation}

The contribution of the unitary term gives 
\begin{equation}
\label{unitarytranssubsub1}
\Delta \mathcal{U}^{(1,0),\mathrm o}_{k,2}(u) = 
-u(u^{2}-1)^{k}(k+1)\left((k/6)(\tilde{u}^{2}-1)P^{(1,1)}_{k-2}(\tilde{u})
+(k+1)P^{(0,0)}_{k}(\tilde{u})\right).
\end{equation}
This formula is obtained by inserting~\eqref{subleadingpochhammer}
into~\eqref{transmissionexample}, by shifting $n \to n-1$, and setting
$a=1$ and $b=(u-1)n$.  

Since the numerator of $\tilde u$ is quadratic
in $u$, the monomials of odd degree in 
\[
(u+1)^{2k+3}\Delta
\mathcal{S}_{k,2}^{(1,0)}(u)
\]
can be extracted from~\eqref{eq:sc2te} by simple inspection; then,
Eqs.~\eqref{oddmonomialsrule} and~\eqref{unitarytranssubsub1}
give
\begin{subequations}
\label{notfromgen}
\begin{equation}
\label{oddmonomialsfromRMT}
\begin{split}
\frac{\Delta \mathcal{T}^{(1,0),\mathrm o}_{k,2} (u)}{u(u^2-1)^{k-2}}
&= -(k/2)(1-2k+2u^{2}+2ku^{2}+u^{4})P^{(1,1)}_{k-2}(\tilde{u})\\
& \quad + (u^{4}-1)k(k-1)P^{(0,1)}_{k-2}(\tilde{u})
- k(k+1)(1-u^{2})^{2}P^{(0,1)}_{k-1}(\tilde{u})\\
&\quad  -(1/2)(u^{2}-1)(k+1)(2k+1-u^{2})P^{(1,1)}_{k-1}(\tilde{u})\\
&\quad -(u^{2}-1)^{2}(k+1)\left((k/6)(\tilde{u}^{2}-1)P^{(1,1)}_{k-2}(\tilde{u})
+P^{(0,0)}_{k}(\tilde{u})\right).
\end{split}
\end{equation}

The contribution of $\Delta \mathcal{T}^{(1,0),\mathrm e}_{k,2}(u)$ can
be computed in the same way.  We obtain
\begin{equation}
  \label{eq:emon}
   \begin{split}
  \frac{\Delta \mathcal{T}^{(1,0),\mathrm e}_{k,2}
    (u)}{u^2(u^2-1)^{k-2}} & =
  (k-1)\bigl((k/2)(\tilde{u}+1)P^{(1,2)}_{k-3}(\tilde u)+u^{2}
     P^{(0,2)}_{k-2}(\tilde u)-(k/3)P^{(1,0)}_{k-2}(\tilde u)\bigr) \\
&\quad +(k/6)(5u^{2}-1)
  P^{(1,1)}_{k-2}(\tilde u) +(2k+1)(u^{2}-1)P^{(0,1)}_{k-1}(\tilde u) \\
& \quad   +(2k/3)(k-1)u^{2}P^{(0,1)}_{k-2}(\tilde u) -
2k(k-1)(1-u^{2})P^{(0,1)}_{k-2}(\tilde u)\\
& \quad - (k+1)^{2}(1-u^{2})P^{(1,1)}_{k-1}(\tilde u) - 
k(k+1)u^{2}P^{(1,1)}_{k-2}(\tilde u)\\
& \quad + k(k-1)P^{(1,1)}_{k-2}(\tilde u).
\end{split}
\end{equation}
\end{subequations}
Inserting Eqs.~\eqref{oddmonomialsfromRMT} and~\eqref{eq:emon}
into~\eqref{eq:totmom2} gives the explicit expression of the second
corrections of the moments.  We now need to demonstrate that they
coincide with the Taylor coefficients of the right-hand side
of~\eqref{eq:gen_diff}.

By comparing~\eqref{eq:gen_diff} with $\Delta \mathcal T_2^{(2,0)}(u,s)$
in Eq.~\eqref{unitaryandreevgenfn} we write 
\begin{equation}
\begin{split}
\Delta \mathcal T^{(1,0)}_2(u,s) & = 3\Delta
\mathcal{T}^{(2,0)}_2(u,s)-\frac{4(u+1)^2}{u}
\Delta \mathcal{T}^{(2,0)}_2(u,s)\\
& \quad +\frac{s(u+1)^2}{u}
\Delta \mathcal{T}^{(2,0)}_2(u,s)
-\frac{(u+1)^2}{u s}\Delta \mathcal{T}^{(2,0)}_2(u,s).
\end{split}
\end{equation}
From Proposition~\ref{unitarysubsubbeta2prop} we know the Taylor
coefficients of $\Delta \mathcal{T}^{(2,0)}_2(u,s)$.  Thus, we end up
with the formulae
\begin{subequations}
\label{fromgen}
\begin{equation}
\label{oddmonomialsfromgen}
\begin{split}
\frac{\bigl[s^k\bigr]^{\mathrm o}\Delta \mathcal T^{(1,0)}_2(u,s)}%
{u(u^{2}-1)^{k}}&= (1/6)(k+2)P^{(1,1)}_{k}(\tilde{u})-(2/3)(k+1)(k+2)
  P^{(0,0)}_{k}(\tilde{u}) \\
& \quad  -(k/6)P^{(1,1)}_{k-2}(\tilde{u})  +(2k/3)(k-1)P^{(0,0)}_{k-2}(\tilde{u})
\\
& \quad  -(k/2)(k+1)(\tilde{u}-1)^{2}u^{2}P^{(1,1)}_{k-2}(\tilde{u})
\end{split}
\end{equation}
for the component including only monomials of odd powers of $u$ and
\begin{equation}
\label{evenmonomialsfromgen}
\begin{split}
\frac{\bigl[s^k\bigr]^{\mathrm e}\Delta \mathcal T^{(1,0)}_2(u,s)}%
{u^{2}(u^{2}-1)^{k-1}}&= (2/3)(k+2)(k+1)P^{(1,1)}_{k-1}(\tilde{u})
-(1/2)(k+1)P^{(1,1)}_{k-1}(\tilde{u}) \\
&\quad +2k(k+1)P^{(0,0)}_{k-1}(\tilde{u})-(2/3)(k-1)kP^{(1,1)}_{k-3}(\tilde{u})
\end{split}
\end{equation}
\end{subequations}
for the monomials of even degree.

Finally, the equivalence of the pair of Eqs.~\eqref{notfromgen}
to~\eqref{fromgen} is proved systematically through steps (1)--(4)
discussed in the proof of Proposition~\ref{delaysubleadingtwoprop}.
\end{proof}
\begin{remark}
  Formula~\eqref{transgensubsub1diff} coincides with the generating
  function computed with semiclassical techniques by Berkolaiko and
  Kuipers.\cite{BK11}
\end{remark}
The computation of the second correction for
Andreev billiards ($\delta \neq 0$) is almost identical to the proof of
Proposition~\ref{transsubsub1prop}, but obviously the algebraic
manipulations are much more involved. We did not carry out the proof in
its entirety.  The expression of such a generating function, however,
can be inferred from general methods and by inspecting lower order
terms. 
\begin{conjecture}
  \label{co:gfgab}
  The generating function of the second correction of the moments of
  the transmission eigenvalues for Andreev billiards with $\beta=1$
  and $\delta > -2$ is 
\begin{equation*}
\begin{split}
\mathcal T^{(1,\delta)}_{2}(u,s) &= -\frac{us\bigl(s^{2}(u-1)^{2}
+3us-(u+1)^{2}\bigr)}%
{\bigl((u+1)^{2}-s(u-1)^{2}\bigr)^{5/2}(1-s)^{3/2}}+\frac{3\delta us}{2(u+1)^{3}(s-1)}\\
&\quad +\delta s\frac{\bigl((u+1)^{2}-s(u-1)^{2}\bigr)^{2}
+2u\delta\bigl((u+1)^{2}-s(u-1)^{2}\bigr)}%
{2\bigl((u+1)^{2}-s(u-1)^{2}\bigr)^{5/2}(1-s)^{3/2}}\\
&\quad +\delta
s\frac{(u+1)^{2}(u^{2}-5u+1)-s(u-1)^{2}(u^{2}-4u+1)}%
{2(s-1)(u+1)\bigl((u+1)^{2}-s(u-1)^{2}\bigr)^{2}}\\
& \quad - \delta \frac{3s^{3}u(u-1)^{4}}%
{2(s-1)(u+1)^{3}\bigl((u+1)^{2}-s(u-1)^{2}\bigr)^{2}}.
\end{split}
\end{equation*}
\end{conjecture}

\section{Asymptotics of Selberg-like Integrals}
\label{se:selberg_like}
In this section we will be interested in asymptotics as $n \to \infty$
of the integrals
\begin{equation}
\label{transmissionintegral2}
\begin{split}
\mathcal{M}_{k}^{(\beta)}(u,v) & = 
\frac{1}{C}\int_{0}^{1}\dotsi \int_{0}^{1}
\left(\sum_{j=1}^{n}x_{j}^{k}\right)
\prod_{j=1}^{n}x_{j}^{\beta/2(b+1)-1}
(1-x_{j})^{\beta/2(a+1)-1} \\
& \quad \times \prod_{1 \leq j < k \leq n}\left
   \lvert x_{k}-x_{j}\right \rvert^{\beta}dx_{1}\dotsm dx_{n}, \quad k
 =1,2,\dotsc
\end{split}
\end{equation}
with the assumption that $\beta \in \{1,2,4\}$ and 
\begin{equation}
\label{eq:absc}
a=(v-1)n \quad \text{and} \quad b=(u-1)n, \quad u,v\ge 1.
\end{equation}

The integrals~\eqref{transmissionintegral2} have been referred to
\emph{Selberg-like} by several authors.\cite{Nov10,Kra10,CDLV10}
Indeed, the normalization constant
\begin{equation}
\label{eq:selberg_int}
\begin{split}
C_n(\beta,a,b)& =\int_{0}^{1}\dotsi \int_{0}^{1}
\prod_{j=1}^{n}x_{j}^{\beta/2(b+1)-1}(1-x_{j})^{\beta/2(a+1)-1}\\
& \quad \times \prod_{1 \leq j < k \leq n}\left
   \lvert x_{k}-x_{j}\right \rvert^{\beta}dx_{1}\dotsm dx_{n}
\end{split}
\end{equation}
was evaluated for any $\beta>0$ by Selberg in 1944,\cite{Sel44} who
obtained the formula
\begin{equation}
C_n(\beta,a,b)= \prod_{j=0}^{n-1}\frac{\Gamma\bigl(\beta/2(b+1+j)\bigr)
\Gamma\bigl(\beta/2(a+1+j)\bigr)\Gamma\bigl(1+(j+1)\beta/2\bigr)}%
{\Gamma\bigl(\beta/2(a+b+1+n+j)\bigr)\Gamma(1+\beta/2)}.
\end{equation}
The integral~\eqref{eq:selberg_int} has since been named
\emph{Selberg's integral}.

In the first part of this work\cite{MS11a} we gave finite-$n$ formulae
for~\eqref{transmissionintegral2} which imply the existence of an
asymptotic expansion in powers of $1/n$
\begin{equation}
\label{doublescaleasymptoticexpansion}
\frac{1}{n}\mathcal{M}^{(\beta)}_{k}(u,v) \sim 
\sum_{p=0}^{\infty}\mathcal{M}^{(\beta)}_{k,p}(u,v)n^{-p}, \quad n \to
\infty.
\end{equation} 
We now compute the first two coefficients in this expansion.

Expressions for the $\beta$-independent leading order term appeared in
Refs.~\citenum{CDLV10} and~\citenum{Kra10} with the assumption that
$\beta=2$; a different formula is also available in the
literature,\cite{Nov10} whose proof was valid for any $\beta$. We
introduce a different approach, which gives a simpler characterization
of the leading order term and allows to compute the next to leading
order coefficient.

\subsection{Leading Order Term}
\label{se:lead_order_selberg}
As for the moments of the transmission eigenvalues, we introduce the
differences 
\begin{subequations}
\label{eq:diff_selberg_like}
\begin{align}
\Delta \mathcal{M}_{k}^{(\beta)}(u,v) & :=
\mathcal{M}_{k}^{(\beta)}(u,v)-\mathcal{M}_{k+1}^{(\beta)}(u,v);
\intertext{similarly, for the coefficients of their asymptotic series
we write}
\Delta \mathcal{M}_{k,p}^{(\beta)}(u,v) & :=
\mathcal{M}_{k,p}^{(\beta)}(u,v)-\mathcal{M}_{k+1,p}^{(\beta)}(u,v).
\end{align}
\end{subequations}
Furthermore, we shall set
\begin{equation*}
  \mathcal{M}^{(\beta)}_{0}(u,v)=  \Delta
  \mathcal{M}^{(\beta)}_{0}(u,v) =0 .
\end{equation*}

\begin{proposition}
  Let $\beta \in \{1,2,4\}$, $a=(v-1)n$ and $b=(u-1)n$.   We have
\begin{equation}
\label{jacobileadingdifferences}
\mathcal{M}_{k,0}^{(\beta)}(u,v)= \frac{u}{u+v} - \sum_{j=1}^{k-1}\frac{1}{j}
\sum_{i=0}^{j}\binom{j}{i}\binom{j}{i-1}\frac{v^{i}u^{j-i+1}
  (u+v-1)^{j-i+1}}{(u+v)^{2j+1}}.
\end{equation}
\end{proposition}
\begin{proof}
  It is sufficient to consider $\beta=2$ because, as for the moments
  of the transmission eigenvalues, a simple calculation discussed in
  Remark~\ref{re:beta_in} shows that the
  terms in (\ref{orthogonaljacobiexample}) that do not come from the
  unitary contributions contribute at subleading order.

 The first term in~\eqref{jacobileadingdifferences} is the limit of
 Aomoto's integral, namely
 \begin{equation*}
  \lim_{n \to \infty}n^{-1}M^{(\beta)}_{\J}(1,n) =\lim_{n \to
    \infty}\frac{b + n}{a + b + 2n} = \frac{u}{u+v}.
 \end{equation*}

 Inserting the leading order term of the asymptotic
 expansion~\eqref{fullgammaratio} into~\eqref{coefficientexample}
 gives
\begin{equation}
\label{leadingestimatejac}
\begin{split}
U^{a,b}_{n,k,j} &= \frac{\bigl((u+v)n-2j+k+1\bigr)\bigl((u+v-1)n\bigr)_{(k-j+1)}
(vn-j+1)_{(j)}(un)_{(k-j+1)}}{\bigl((u+v)n-j\bigr)_{(k+2)}
\bigl((u+v)n-j+1\bigr)_{(k)}(n+1)_{(-j)}}\\
& \sim n\frac{v^{j}u^{k-j+1}(u+v-1)^{k-j+1}}{(u+v)^{2k+1}}+O(1), 
\quad n \to \infty, 
\end{split}
\end{equation}
where $a$ and $b$ depend on $u$ and $v$ through the
scaling~\eqref{eq:absc}. Substituting the right-hand side
of~\eqref{leadingestimatejac} into~(\ref{transmissionexample}) 
leads to a formula for the differences:
\begin{equation}
\label{eq:lead_diff_s_int}
\Delta \mathcal{M}_{k,0}^{(\beta)}(u,v)=\frac{1}{k}\sum_{j=0}^{k}\binom{k}{j}
\binom{k}{j-1}\frac{v^{j}u^{k-j+1}
  (u+v-1)^{k-j+1}}{(u+v)^{2k+1}}.
\end{equation}
\end{proof}

\subsection{Next to Leading Order Term}
\begin{proposition}
\label{pr:ntlsi}
Let $\beta \in \{1,2,4\}$, $a=(v-1)n$ and $b=(u-1)n$. The next to
leading order coefficient in the
expansion~\eqref{doublescaleasymptoticexpansion} is 
\begin{equation*}
\mathcal{M}_{k,1}^{(\beta)}(u,v) = \left(\frac{2}{\beta}-1\right)
\frac{v^{k}}{2(u+v)^{2k}}\sum_{j=0}^{k}
\left(\binom{2k}{2j}-\binom{k}{j}^{2}\right)\left(\frac{u(u+v-1)}{v}\right)^{j}.
\end{equation*}
\end{proposition}
\begin{proof}
  We discuss only the case when $\beta=1$, as when $\beta=2$ and
  $\beta=4$ the proof is analogous.

  We can write down a relation analogous to that one
  in~(\ref{transasymptrule}), where now $a=(v-1)n$ grows with $n$
  instead of being independent of $n$. For the first subleading
  correction it reads
\begin{equation}
\label{jacb1asymptrule}
  \mathcal{M}_{k,1}^{(1)}(u,v) = \mathcal{U}_{k,1}^{\J}(u,v) -
  \mathcal{S}^{\J}_{k,1}(u,v)+\mathcal{I}^{\J}_{k,1}(u,v).
\end{equation}
Insert the leading order of~\eqref{fullgammaratio} into the
symplectic contribution in~\eqref{orthogonaljacobiexample}; then, a
calculation similar to the derivation of (\ref{subleadingdoublesum})
in Proposition~\ref{pr:ntl_teig} shows that
\begin{equation*}
\begin{split}
  \mathcal{S}^{\J}_{k,1}(u,v) &= \sum_{j=1}^{\lfloor k/2 \rfloor}
\sum_{i=0}^{k-2j}\binom{k}{i}\binom{k}{i+2j}
\frac{v^{i+j}u^{k-i-j}(u+v-1)^{k-i-j}}{(u+v)^{2k}}\\
  &=\frac{v^{k}}{2(u+v)^{2k}}\sum_{j=0}^{k-1}\left(\binom{2k}{2j}-\binom{k}{
      j}^{2}\right)\left(\frac{u(u+v-1)}{v}\right)^{j}.
\end{split}
\end{equation*}
In the same way, starting from formula \eqref{eq:Iab} we arrive at
\begin{equation*}
\mathcal{I}^{\J}_{k,1}(u,v) = \frac{v^{k}}{(u+v)^{2k}}
\sum_{j=0}^{k}\binom{2k}{2j}\left(\frac{u(u+v-1)}{v}\right)^{j}.
\end{equation*}
Therefore, it is sufficient to prove that
\begin{equation*}
  \mathcal{U}_{k,1}^{\J}(u,v) = -\frac{v^{k}}{(u+v)^{2k}}
\sum_{j=0}^{k}\binom{k}{j}^{2}\left(\frac{u(u+v-1)}{v}\right)^{j} 
= -\frac{v^{k}}{(u+v)^{2k}}(y-1)^{k}P^{(0,0)}_{k}(\tilde{y}),
\end{equation*}
where we have introduced the new variables
\begin{equation}
y = \frac{u}{v}(u+v-1) \quad \text{and} \quad  \tilde{y} = \frac{y+1}{y-1}.
\end{equation}

Taking the differences, we see that the statement of the proposition
is equivalent to
\begin{equation}
\label{jacsubleadingfinal}
\Delta \mathcal{U}_{k,1}^{\J}(u,v) = 
-\frac{v^{k}(y-1)^{k}}{(u+v)^{2k}}P^{(0,0)}_{k}(\tilde{y})
+\frac{v^{k+1}(y-1)^{k+1}}{(u+v)^{2k+2}}P^{(0,0)}_{k+1}(\tilde{y}).
\end{equation}
% The task, then, is to obtain an expression for $\Delta
% \mathcal{U}_{\mathrm{J},k}^{1}(u,v)$ from the exact result
% (\ref{transmissionexample}).  To this end, as per equa

Replace $n \to n-1$ in (\ref{transmissionexample}) and insert the
scaling~\eqref{eq:absc}. Then, use (\ref{subleadingpochhammer}) to
expand~\eqref{jacsubleadingfinal} up to the next to leading order
correction.  We obtain
\begin{equation}
\label{jacsubleadingexpr}
\begin{split}
\Delta \mathcal{U}^{\J}_{k,1}(u,v) & = \frac{1}{k}
\sum_{j=1}^{k}\binom{k}{j}\binom{k}{j-1}
\frac{u^{k-j+1}v^{j}(u+v-1)^{k-j+1}}{2(u+v)^{2k+1}}\\
& \quad \times \left(\frac{(k-j+1)(k-j-2)}{(u+v-1)} 
   -\frac{j(j+1)}{v}+\frac{(k-j+1)(k-j-2)}{u}\right.\\
& \quad \left.+\frac{6k+4-2k^{2}+4kj}{(u+v)}-j(1+j)\right),
\end{split}
\end{equation}
which may be written in terms of Jacobi polynomials as 
\begin{equation}
\label{jacsubleadingpolys}
\begin{split}
\Delta \mathcal{U}^{\J}_{k,1}(u,v) &= \frac{(y-1)^{k-2}u(u+v-1)v^{k}}%
{(u+v)^{2k+1}}\left(\left(\frac{2k}{u+v}
-\frac{1}{v}-1\right)(y-1)P^{(0,1)}_{k-1}(\tilde{y})\right.\\
& \left. \quad +\frac{k(u-1)}{v}P^{(1,1)}_{k-2}(\tilde{y})
-\left(\frac{1}{u+v-1}+\frac{1}{u}\right)(y-1)P^{(1,0)}_{k-1}(\tilde{y}) \right.\\
& \left.\quad+\frac{(6k+4-2k^{2})(y-1)}{2k(u+v)}P^{(1,1)}_{k-1}(\tilde{y})\right).
\end{split}
\end{equation}

By subtracting~\eqref{jacsubleadingfinal}
from~\eqref{jacsubleadingpolys} we obtain an expression involving
Jacobi polynomials.  In order to complete the proof, we need to show
that it vanishes.  As previously, we achieve it systematically,
following the steps (1)--(4) outlined at the end of the proof of
Proposition~\ref{delaysubleadingtwoprop}. 
\end{proof}

\subsection{Leading Order Generating Function}
\label{sse:lead_ord_gen_fun}

Let us define the generating function
\begin{equation}
  \label{eq:sel_int_gf}
  H(u,v;s) = \sum_{k=0}^{\infty}\mathcal{M}_{k,0}^{(\beta)}(u,v)s^k.
\end{equation}
We now present a new proof of the following.
\begin{theorem}[Novaes\cite{Nov10}]
  The generating function $H(u,v;s)$ obeys the quadratic equation
  \begin{equation}
     \label{lastquadraticeqn}
      H = \frac{us}{u+v-(1+u)s}-\frac{1-s}{u+v-(1+u)s}H^{2}.
  \end{equation}  
\end{theorem}
\begin{proof}
  By rearranging summation indices as $j \to k-j$, we can write
  formula (\ref{jacobileadingdifferences}) in terms of a  Narayana
  polynomial of degree $k$ (see~\eqref{eq:narayanapolynomial}):
\begin{equation*}
\begin{split}
\Delta \mathcal{M}_{k,0}^{(\beta)}(u,v) &=
\frac{v^{k+1}}{(u+v)^{2k+1}}
\frac{1}{k}\sum_{j=0}^{k}\binom{k}{j}\binom{k}{j-1}
\left(\frac{u(u+v-1)}{v}\right)^{j}\\
&=\frac{v^{k+1}}{(u+v)^{2k+1}}N_{k}\left(\frac{u(u+v-1)}{v}\right).
\end{split}
\end{equation*}
Thus, the generating function of $\Delta
\mathcal{M}^{(\beta)}_{k,0}(u,v)$ can be easily related to the
generating function of the Narayana polynomials, which we report
in~\eqref{narayanagenerator}.  More explicitly
\begin{equation}
\label{eq:diff_gf}
\widetilde{H}(u,v;s) := \sum_{k=1}^{\infty}\Delta 
\mathcal{M}_{k,0}^{(\beta)}(u,v)s^{k} 
= \frac{v}{u+v}\rho\left(\frac{u(u+v-1)}{v},\frac{sv}{(u+v)^{2}}\right).
\end{equation}

The generating functions $H$ and $\widetilde{H}$ are related by
\begin{equation}
  \label{eq:diff_gf_rel}
  \widetilde{H} = \left(1 - s^{-1}\right)H +
  \mathcal{M}^{(\beta)}_{1,0}(u,v). 
\end{equation}
It follows immediately that 
\begin{equation}
     \label{genfndoublescale}
      H  = \frac{s}{s-1}\left(\frac{v}{u+v}
\rho\left(\frac{u(u+v-1)}{v},\frac{sv}{(u+v)^{2}}\right)-\frac{u}{u+v}\right)
  \end{equation}

  The generating function $\rho(x,s)$ satisfies the quadratic equation
  (see, \textit{e.g.}, Ref.~\citenum{Deu99})
\begin{equation*}  
   \rho = (x+\rho)(1+\rho)s,
\end{equation*}
which leads to
\begin{equation*}
(us+(s-1)H)(s+(s-1)H) = (u+v)(s-1)H+us.
\end{equation*}
Rearranging this expression gives~\eqref{lastquadraticeqn}.
\end{proof}

Formula (\ref{lastquadraticeqn}) is equivalent (in our notation) to
Eq.~(58) in Ref.~\citenum{Nov10} obtained by a completely different
approach.

\subsection{The Limiting Eigenvalue Density of the Jacobi Ensembles}
\label{sse:limiting_density}
The leading order term of the density~\eqref{meanevdensity} averaged
over the Jacobi ensembles first appeared in the multivariate
statistics literature\cite{Wac80,BYK87} and later in the context of
chaotic quantum transport, first when $a=0$ and $b=0$\cite{Bee93,BM94}
and then only when $a=0$.\cite{BB96} More recently it has attracted
attention in the free probability literature\cite{CC04,Col05} with the
same scaling as that adopted in this section.  It was found that
\begin{equation}
\label{limitingjacobiesd}
\rho_{\infty}(x) = \lim_{n \to \infty}n^{-1}\rho_{n}(x) = 
(u + v)\frac{\sqrt{(x-\lambda_{-})(\lambda_{+}-x)}}{2\pi x(1-x)},
\end{equation}
with support
\begin{equation}
\label{lambda}
\lambda_{\pm} = \left(\sqrt{\frac{u}{u+v}\left(1-\frac{1}{u+v}\right)} 
\pm \sqrt{\frac{1}{u+v}\left(1-\frac{u}{u+v}\right)}\right)^{2}.
\end{equation}
The density~\eqref{limitingjacobiesd} is normalized to one; in
Refs.~\citenum{CC04} and \citenum{Col05} a different normalization and
notation is used.  We now show that the leading order contributions of
the Selberg-like integrals $\mathcal{M}_{k}^{(\beta)}(u,v)$ are the
moments of the limiting density.

Formula~\eqref{limitingjacobiesd} makes it apparent why it is helpful
to consider the differences between the moments, rather than the
moments themselves: multiplication by 
\[
x^{k}-x^{k+1}=(1-x)x^{k}
\]
cancels the poles in the denominator of \eqref{limitingjacobiesd},
making the resulting integration tractable.

Our method is similar in style to that used in Ref.~\citenum{HT03} for a
similar problem involving the Laguerre ensembles.

\begin{theorem}
The difference of the moments of the limiting density
$\rho_{\infty}(x)$ are 
\begin{equation}
\label{eq:limit_diff_2}
\begin{split}
\int_{\lambda_{-}}^{\lambda_{+}}x^{k}(1-x)\rho_{\infty}(x)dx &=
\frac{1}{k}\sum_{j=0}^{k}\binom{k}{j} \binom{k}{j-1}\frac{v^{j}u^{k-j+1}
  (u+v-1)^{k-j+1}}{(u+v)^{2k+1}}\\
&  = \Delta \mathcal{M}_{k,0}^{(\beta)}(u,v).
\end{split}
\end{equation}
\end{theorem}
\begin{proof}
Define
\begin{equation}
\label{eq:diff_mom_def}
\begin{split}
  I_{k} & =
  \int_{\lambda_{-}}^{\lambda_{+}}x^{k}(1-x)\rho_{\infty}(x)dx\\
  & = \frac{u+v}{2\pi}\int_{\lambda_{-}}^{\lambda_{+}}x^{k-1}\left(
    \sqrt{x(\lambda_{-}+\lambda_{+})
      -\lambda_{-}\lambda_{+}-x^{2}}\right)dx.
\end{split}
\end{equation}
The substitution 
\[
x = \frac{\lambda_{-}+\lambda_{+}}{2} + 
\frac{\lambda_{+}-\lambda_{-}}{2}\cos(\theta)
\]
transforms the integral in the right-hand side
of~\eqref{eq:diff_mom_def} into
\begin{equation*}
I_{k} = \frac{\alpha_{2}^{2}(u+v)}{4\pi}
\int_{-\pi}^{\pi}\bigl(\alpha_{1}+\alpha_{2}\cos(\theta)\bigr)^{k-1}
\sin^{2}(\theta)d\theta,
\end{equation*}
where $\alpha_{1} = \frac{\lambda_{+}+\lambda_{-}}{2}$ 
and $\alpha_{2} = \frac{\lambda_{+}-\lambda_{-}}{2}$. 

Let $p$ and $q$ and $z(\theta)$ be such that
\begin{equation}
\label{defsabz}
\alpha_{1} = p^{2} + q^{2}, \quad \alpha_{2} = 2pq, \quad 
z(\theta) = (p+qe^{i\theta})^{k-1}. 
\end{equation}
Then, we have
\[
\left|z(\theta)\right|^{2} = \bigl(\alpha_{1}+\alpha_{2}\cos(\theta)\bigr)^{k-1}.
\] 
Set
\[
\zeta = e^{i\theta}z(\theta) \quad \text{and} \quad  
\eta = e^{-i\theta}z(\theta).
\] 
Using 
\[
\sin^{2}(\theta) = \frac{1-\cos(2\theta)}{2} 
= \frac{1}{2}\Re{\bigl(1-e^{2i\theta}\bigr)}
\] 
we see that
\begin{equation*}
\begin{split}
\frac{I_{k}}{u+v} &= \frac{\alpha_{2}^{2}}{8\pi}
\int_{-\pi}^{\pi}\Re{(1-e^{2i\theta})}|z(\theta)|^{2}d\theta\\
&=\frac{\alpha_{2}^{2}}{8\pi}\left(\int_{-\pi}^{\pi}\left|z(\theta)\right|^{2}d\theta
-\Re{\left(\int_{-\pi}^{\pi}\zeta(\theta)
\overline{\eta(\theta)}d\theta\right)}\right).
\end{split}
\end{equation*}
These integrals can be evaluated by calculating the Fourier modes of
$\zeta$, $\eta$ and $z$ and applying Parseval's identity. The Fourier
coefficients are
\begin{equation*}
z_{j} = p^{k-1-j}q^{j}\binom{k-1}{j}, \quad \zeta_{j} 
= p^{k-j}q^{j-1}\binom{k-1}{j-1}, \quad \eta_{j} 
= p^{k-j-2}q^{j+1}\binom{k-1}{j+1}.
\end{equation*}
Parseval's identity gives
\begin{equation}
\label{jacldensderiv}
\begin{split}
\frac{I_{k}}{u+v} &= \frac{\alpha_{2}^{2}}{4}
\sum_{j=0}^{k-1}p^{2k-2j-2}q^{2j}\binom{k-1}{j}^{2}
-\frac{\alpha_{2}^{2}}{4}\sum_{j=0}^{k-1}p^{2k-2j-2}q^{2j}
\binom{k-1}{j-1}\binom{k-1}{j+1} \\
&= \frac{\alpha_{2}^{2}}{4k}\sum_{j=0}^{k-1}p^{2k-2j-2}q^{2j}
\binom{k}{j+1}\binom{k}{j}.
\end{split}
\end{equation}
Solving $p$ and $q$ in terms of $\alpha_{1}$ and $\alpha_{2}$ 
in Eq.~\eqref{defsabz} and using~\eqref{lambda} allows us to
write 
\begin{equation*}
p = \sqrt{\frac{u(u+v-1)}{(u+v)^{2}}}, \quad 
q = \sqrt{\frac{v}{(u+v)^{2}}}, \quad
 \alpha_{2}^{2} = \frac{4uv(u+v-1)}{(u+v)^{4}}.
\end{equation*}
Substituting these expressions into \eqref{jacldensderiv}
and replacing $j \to j-1$ gives Eq.~\eqref{eq:limit_diff_2}.
\end{proof}

\section{Conclusions and Outlook}
\label{se:conclusions}
In this article we have computed the first three terms of the
asymptotic expansions in the limit as $n \to \infty$ of the moments of
the density of the transmission eigenvalues and delay times in chaotic
ballistic cavities.  The fundamental assumption is that the chaotic
dynamics allows us to model the scattering
matrix~\eqref{eq:landauer_buttiker} with the circular ensembles from
RMT.  Our formulae are available for all symmetry classes $\beta \in
\{1,2,4\}$ --- with a few exceptions for the second corrections to the
leading order terms, for which we did not perform the calculations for
$\beta=4$.  For the moments of the transmission eigenvalues we treat
Andreev billiards too.  Finally, we studied the asymptotics of the
Selberg-like integrals as well.

Our results on the moments of the transmission eigenvalues and proper
time delays for $\beta=1$ and $\beta=2$ symmetry classes agree with
those computed using semiclassical
techniques.\cite{BHN08,BK10,BK11,BK12,Nov11} It would be interesting
to know if the results that we obtain for Andreev billiards beyond the
leading order agree with semiclassics too.

In recent announcements Berkolaiko and Kuipers\cite{BK12} and
independently Novaes\cite{Nov11} show that semiclassical computations
lead to the same asymptotic expansions of the moments of the
transmission eigenvalues as RMT. Their results involve combinatorial
expressions for the correlations of the scattering trajectories.  At
this stage, however, it is not clear how to extract explicit formulae
from the combinatorics.  Their formulae include unsolved combinatorial
problems too.  It is a challenging project to go even further and to
obtain the full asymptotic expansions for the moments of the
transmission eigenvalues and proper delay times using RMT
techniques. Nevertheless, it would be interesting pursue this program,
since combinatorics would not appear in the calculations and in the
final formulae.  Furthermore, the equivalence of the two approaches
could answer unsolved combinatorial problems.

In principle such asymptotic series can be obtained because the
asymptotic properties discussed in Sec.~\ref{ss:examp} hold at all
orders in negative powers of $n$. The finite-$n$ formulae that we
computed in the first part of this project\cite{MS11a} contain
products of the form
\begin{equation}
\label{productgamma}
\prod_{j=1}^{p}\frac{\Gamma(a_{j}z+c_{j})}{\Gamma(b_{j}z+d_{j})}.
\end{equation}
One could insert the asymptotic series~\eqref{fullgammaratio}
into~\eqref{productgamma} and repeat the systematic procedures
developed in this article to find further corrections.  In practise,
however, when we go beyond the second correction it is not clear how
to organise the resulting asymptotic terms in a way that would lead to
results in closed form, as those obtained in this paper.
 
We believe that, as for the asymptotic of the ratio
(\ref{fullgammaratio}), it may be possible to establish recursion
relations for the terms in the asymptotic series of
(\ref{productgamma}). This would lead to recursion relations for the
full expansions of the moments of the proper delay times
and of the transmission eigenvalues.  

\appendix
\numberwithin{equation}{section}
\section{Exact Results for $\beta=1$ and $\beta=4$ 
Matrix Ensembles}
\label{sec:exactresults}

In this appendix we report formulae for finite-$n$ that we derived in
the first part of this work\cite{MS11a} and that are the starting
points for our asymptotic analysis.

Equations~(\ref{orthogonaljacobiexample}) and
(\ref{orthogonalwishartexample}) for $\beta=1$ were expressed in terms
of certain coefficients which are defined below. We also give the
corresponding results for $\beta=4$. The expressions for $\beta=2$ are
presented at equations (\ref{transmissionexample}) and
(\ref{wignerexample}). Recall that
\[
(x)_{(n)} =
\frac{\Gamma(x+n)}{\Gamma(x)}
\]
is the  Pochhammer symbol.

In the following we assume that the channel number satisfies
$n>k\beta/2$. This allows for some simplifications; furthermore, a
small $n$ is not required for an asymptotic analysis.

\subsection{Moments of the Proper Delay Times}
We need the following expressions for the Laguerre ensemble:
\begin{subequations}
\begin{align}
\label{Sb}
S_{i,j}^{b}(-k,n) &= \frac{(2b+2n)_{(-k-i-2j+1)}(2n-i-2j+1)_{(i)}} 
{2^{-k-2j+2}(n+1)_{(-j)}(b+n)_{(1-j)}},  \\
\label{Ib}
I_{\Lag}(-k,n)& =2^{-k}\sum_{j=0}^{n/2-1}\binom{2k+2j-1}{2j}
\frac{\bigl(\frac{1}{2}(b+n)\bigr)_{(-k-j)}}%
{\bigl(\frac{1}{2}(1+n)\bigr)_{(-j)}} 
+ \phi_{-k,n}^{\mathrm{L}},
\end{align}
\end{subequations}
where $\phi_{-k,n}^{\mathrm{L}}$ is an exponentially decaying term
which has no contribution to the asymptotic expansions
\eqref{eq:asy_ex_td} for any asymptotic order $p$. Its explicit
expression can be found in Ref.~\citenum{MS11a}, Sec. 6.3,
Eq. (6.29). The exact moments for $\beta=1$ are then given by
Eq. (\ref{orthogonalwishartexample}); for $\beta=4$ we have
\begin{equation}
\begin{split}
M_{\Lag}^{(4)}(-k,n) & = 2^{k-1}M_{\mathrm{L}_{2b}}^{(2)}(-k,2n)\\
& \quad -\sum_{j=1}^{n}\sum_{i=0}^{2n-2j}\binom{k+j-1}{k-1}
\binom{k+i+2j-1}{k-1}S_{i,j}^{b}(-k,n).
\end{split}
\end{equation}
\subsection{Moments of the Transmission Eigenvalues}
The formulae we need for the Jacobi ensemble are
\begin{equation}
\label{Sab}
\begin{split}
 S^{a,b}_{i,j}(k,n) & = \frac{2^{4j-3}(2a+2n-i-2j+1)_{(i)}
   (2b+2n)_{(k-i-2j+1)}(2a+2b+2n)_{(k-i-2j+1)}}%
 {(2n-2j+1)_{(-i)}(n+1)_{(-j)}(a+n+1)_{(-j)}(b+n)_{(1-j)}(a+b+n)_{(1-j)}} \\
 &\quad \times
 \frac{(2a+2b+4n-4j+1)(2a+2b+4n-2i-4j+k+1)}%
{(2a+2b+4n-i-2j+1)_{(1+k)}(2a+2b+4n-i-4j+1)_{(1+k)}}.
\end{split}
\end{equation}
and
\begin{equation}
\label{eq:Iab}
\begin{split}
I_{\J}(k,n) = & 4^{k}\sum_{j=0}^{k}\binom{2k}{2j}
\frac{(a+b+2n-4j-1+2k)(\frac{1}{2}(a+b+n))_{(k-j)}(\frac{1}{2}(b+n))_{(k-j)}}
{(a+b+2n-2j-1)_{(2k+1)}(\frac{1}{2}(a+n+1))_{(-j)}(\frac{1}{2}(1+n))_{(-j)}}
\end{split}
\end{equation}

The exact moments for $\beta=1$ are given in
(\ref{orthogonaljacobiexample}); for $\beta=4$ we have 
\begin{equation}
\label{symplecticjacobi}
M_{\J}^{(4)}(k,n) = \frac{1}{2}M_{\mathrm{J}_{2a,2b}}^{(2)}(k,n)
-\sum_{j=1}^{\lfloor k/2 \rfloor}
\sum_{i=0}^{k-2j}\binom{k}{i}\binom{k}{i+2j}S^{a,b}_{i,j}(k,n).
\end{equation}

\section{Jacobi Polynomials and Hypergeometric functions}
\label{jp_hf}

This appendix contains definitions and identities of Jacobi
polynomials and Gauss hypergeometric functions that are needed
throughout the paper.  They can be found in standard references, in
particular in the books by Abramowitz and Stegun,\cite{AS72} Ch.
15, and Szeg\H{o},\cite{Sze39} Ch. 4.  However, because of their
extensive use in many proofs, for the convenience of the reader we list
them here.

The Jacobi polynomials may be defined explicitly by the formula
\begin{equation}
\label{jacdef}
P_{n}^{(\alpha,\beta)}(x) := \sum_{j=0}^{n}\binom{n+\alpha}{n-j}
\binom{n+\beta}{j}\left(\frac{x-1}{2}\right)^{j}
\left(\frac{x+1}{2}\right)^{n-j}.
\end{equation}
If $\alpha > -1$ and $\beta > -1$ they are orthogonal on $[-1,1]$ with
respect to the measure 
\[
d\mu = (1-x)^{\alpha}(1+x)^{\beta}dx
\]
and satisfy the three-term recurrence equation
\begin{multline}
  \label{threetermrec}
  2n(n+ \alpha + \beta)(2n + \alpha + \beta -2)
  P^{(\alpha,\beta)}_n(x)\\
= (2n + \alpha + \beta -1)\bigl((2n + \alpha + \beta)(2n + \alpha +
\beta - 2)x + \alpha^2 - \beta^2\bigr)P^{(\alpha,\beta)}_{n-1}(x)\\
-2(n+\alpha - 1)(n + \beta - 1)(2n + \alpha +
\beta)P^{(\alpha,\beta)}_{n-2}(x), \quad  n =2,3,\dotsc,
\end{multline}
with initial conditions
\begin{equation*}
  P_0^{(\alpha,\beta)}(x)=1 \quad \text{and} \quad
  P_1^{(\alpha,\beta)}(x) = \tfrac12 (\alpha + \beta + 2 )x +
  \tfrac12(\alpha - \beta).
\end{equation*}
They have connection coefficients of the form ($p \in \mathbb{Z}$)
\begin{subequations}
\label{connection1}
\begin{align}
P_{n}^{(\alpha,\beta)}(x) & = \sum_{j=0}^{p}
\mathcal{C}_{j,p,n}^{\alpha,\beta}P_{n-j}^{(\alpha,\beta+p)}(x), \\
P_{n}^{(\alpha,\beta)}(x) & =
\sum_{j=0}^{p}(-1)^{j}
\mathcal{C}_{j,p,n}^{\beta,\alpha}P_{n-j}^{(\alpha+p,\beta)}(x),
\end{align}
\end{subequations}
where
\begin{equation*}
\mathcal{C}_{j,p,n}^{\alpha,\beta}  = \binom{p}{j}
\frac{(\alpha+\beta+n+1)_{(p-j)}(\alpha+\beta+2n-2j+1+p)}%
{(\alpha+\beta+2n-j+1)_{(1+p)}(\alpha + n+1)_{(-j)}}.
\end{equation*}

The Gauss hypergeometric function is defined in the unit circle by the series
\begin{equation}
\label{hypergeomdef}
{}_2F_1(a,b;c;z):= \sum_{j=0}^\infty \frac{(a)_j(b)_j}{(c)_j} \, \frac
{z^j} {j!}; 
\end{equation}
it can be analytically continued in the rest of the complex
plane. For special values of $a$ and $b$ this series truncates and
becomes a polynomial.  In particular, for $\alpha,\beta \in
\mathbb{Z}$ it is related to the Jacobi
polynomials by the identities 
\begin{subequations}
\begin{align}
\label{jacobihypergeom}
{}_2F_1(-n,-n+\beta;\alpha+1;x) & = \frac{\Gamma(\alpha+1)
\Gamma(n-\beta+1)}{\Gamma(n+\alpha-\beta+1)}
(x-1)^{n-\beta}P^{(\beta,\alpha)}_{n-\beta}\left(\frac{x+1}{x-1}\right),\\
\label{wignerhypergeom}
{}_2F_1(n,n+\beta;\alpha+1;x^{-1}) &=
\left(\frac{x}{x-1}\right)^{\beta+n}
\frac{\Gamma(\alpha+1)\Gamma(n-\alpha)}%
{\Gamma(n)}P^{(\alpha,\beta)}_{n-\alpha-1}\left(\frac{x+1}{x-1}\right).
\end{align}
\end{subequations}

% Jacobi polynomials with negative parameters can be expressed
% in terms of Jacobi polynomials with positive parameters via
% \begin{align*}
%  P^{(-l,\beta)}(x) & =\frac{\Gamma(n+\beta+1)}%
% {\Gamma(n+\beta+1-l)}\frac{(n-l)!}{n!}
%   \left(\frac{x-1}{2}\right)^lP^{(l,\beta)}_{n-l}(x),\\
% P^{(\alpha,-l)}_{n}(x) & = \frac{\Gamma(n+\alpha+1)}%
% {\Gamma(n+\alpha+1-l)}\frac{(n-l)!}{n!}
% \left(\frac{x+1}{2}\right)^{l}P^{(\alpha,l)}_{n-l}(x),
% \end{align*}
% where $l$ is an integer such that $1\le l \le n$.

\section{Generalizations of Chu-Vandermonde's summations}
\label{sse:chu_vandermonde}

Many proofs in Sec.~\ref{se:orth_sym} require computing sums of
the form
\begin{subequations}
\begin{gather}
\label{chumomentspos}
\sum_{i=1}^{j}\binom{k}{j-i}\binom{k}{j+i}i^{p}, \\
\label{chumomentsneg}
\sum_{i=1}^{j}\binom{k+j-i-1}{k-1}\binom{k+j+i-1}{k-1}i^{p}, 
\end{gather}
\end{subequations}
where $p$ is an integer. These may be interpreted as generalisations
of the classical Chu-Vandermonde convolution identities
\begin{subequations}
\begin{align}
\label{vandermondeconvpos}
\sum_{i=-m}^{l}\binom{r}{m+i}\binom{s}{l-i} & = \binom{r+s}{m+l}, \\
\label{vandermondeconvneg}
\sum_{i=-m}^{l}\binom{m+i}{r}\binom{l-i}{s}& =\binom{m+l+1}{r+s+1}.
\end{align}
\end{subequations}

We now present two lemmas which show how the asymptotics of the double
sums in Eqs. (\ref{orthogonaljacobiexample}) and
(\ref{orthogonalwishartexample}) are related to sums
(\ref{chumomentspos}) and (\ref{chumomentsneg}) respectively. We then give
two examples describing the basic strategy used to evaluate
(\ref{chumomentspos}) and (\ref{chumomentsneg}).
\begin{lemma}
\label{chupropjacobi}
Set $\theta = \lfloor k/2 \rfloor$. For any finite set of coefficients
$\left \{C_{i,j}\right \}$ we have the identity
\begin{equation}
\label{relationtochu}
\sum_{j=1}^{\theta}\sum_{i=0}^{k-2j}\binom{k}{i}\binom{k}{i+2j}C_{i,j}
= \sum_{j=1}^{k-1}\sum_{i=1}^{j}\binom{k}{j-i}\binom{k}{i+j}C_{k-i-j,i}.
\end{equation}
\end{lemma}
\begin{proof}
  Let $S_k$ denote the left-hand side of
  Eq.~\eqref{relationtochu}. Shifting the summation index in
  the inner sum as $i \to k-i-j$ gives
\begin{equation}
\label{eq:skf}
\begin{split}
S_k& =\sum_{j=1}^{\theta}\sum_{i=j}^{k-j}\binom{k}{k-i-j}
\binom{k}{k-i+j}C_{k-i-j,j}\\
& =\sum_{i=1}^{k-1}
\sum_{j=1}^{\theta -
  |\theta-i|+1}\binom{k}{i+j}\binom{k}{i-j}C_{k-i-j,j},
\end{split}
\end{equation}
where in the last passage we interchanged the order of summation. By
decomposing the outer sum according to whether $\theta-i \leq 0$ or
$\theta-i>0$ Eq.~\eqref{eq:skf} becomes
\begin{equation*}
%\label{lastlinechurelation}
\begin{split}
S_k&=\sum_{i=1}^{\theta}\sum_{j=1}^{i+1}\binom{k}{i+j}
\binom{k}{i-j}C_{k-i-j,j}\\
& \quad + \sum_{i=\theta+1}^{k-1}\sum_{j=1}^{2\theta-i+1}\binom{k}{i+j}
  \binom{k}{i-j}C_{k-i-j,j}.
\end{split}
\end{equation*}
Now note that the upper limits in both inner sums can be replaced by
$i$:  in the first one because if $j=i + 1$, then
\begin{equation*}
  \binom{k}{i-j}=\binom{k}{-1}=0;
\end{equation*}
in the second inner sum $i \ge \theta +1$, thus $2\theta-i + 1\le i$.
Since $j$ lies in the range 
\[
2\theta - i + 1 < j \le i, 
\]
then
\[
j + i> 2\theta + 1>k  \quad \text{and} \quad \binom{k}{i +j}=0,
\]
yielding no contribution to the sum. Finally, relabelling $(i,j) \to
(j,i)$ gives the right-hand side of Eq. (\ref{relationtochu}).
\end{proof}
\begin{lemma}
\label{chuproplaguerre}
Suppose that $\left \{C_{i,j}\right \}$ is any set of coefficients
such that the double series 
\begin{equation}
\label{lagchu1stline}
\Theta= \sum_{j=1}^{\infty}\sum_{i=0}^{\infty}\binom{k+i-1}{k-1}
\binom{k+2j+i-1}{k-1}w^{-k-j-i}C_{i,j}
\end{equation}
is absolutely convergent. Then
\begin{equation}
\label{lagchu2ndline}
\Theta =\sum_{j=1}^{\infty}\sum_{i=1}^{j}
\binom{k+j-i-1}{k-1}\binom{k+j+i-1}{k-1}w^{-k-j}C_{j-i,i}.
\end{equation}
\end{lemma}
\begin{proof}
  First, shift the summation index $i \to i-j$ in
  (\ref{lagchu1stline}); then, (\ref{lagchu2ndline}) is obtained by
  interchanging the order of summation and relabelling the indices.
\end{proof}

Lemmas \ref{chupropjacobi} and \ref{chuproplaguerre} are important
because the inner sums on the right-hand side of (\ref{relationtochu})
and in (\ref{lagchu2ndline}) can be explicitly computed whenever
$C_{i,j}$ is a polynomial in $i$ and $j$.  This technique is better
illustrated with two examples.

\begin{example}
\label{ex:chu_sum1}
Set $p=0$ in (\ref{chumomentsneg}). From (\ref{vandermondeconvneg})
we have
\begin{equation*}
\sum_{i=-j}^{j}\binom{k+j-i-1}{k-1}\binom{k+j+i-1}{k-1} 
= \binom{2k+2j-1}{2k-1},
\end{equation*}
which implies
\begin{equation*}
\sum_{i=1}^{j}\binom{k+j-i-1}{k-1}\binom{k+j+i-1}{k-1} 
= \frac{1}{2}\left(\binom{2k+2j-1}{2k-1}-\binom{k+j-1}{k-1}^{2}\right).
\end{equation*}
This formula was used to obtain (among others)
Eq. (\ref{lastlinederiv1}) in
Proposition~\ref{subleadingwignerdelayrmt}.
\end{example}
\begin{example}
\label{ex:chumom} 
Now take $p=2$ in the sum \eqref{chumomentspos}. We have
\begin{equation*}
\begin{split}
\sum_{i=-j}^{j}\binom{k}{j-i}\binom{k}{j+i}i^{2} & = k\sum_{i=-j}^{j}
\binom{k}{j-i}\binom{k-1}{j+i-1}i
   -j\sum_{i=-j}^{j}\binom{k}{j-i}\binom{k}{j+i}i\\
 &  =-k^{2}\sum_{i=-j}^{j}\binom{k-1}{j-i-1}
  \binom{k-1}{j+i-1}+kj\sum_{i=-j}^{j}\binom{k-1}{j-i-1}\binom{k}{j+i}\\
&\quad -j^{2}\sum_{i=-j}^{j}\binom{k}{j-i}\binom{k}{j+i}
   +kj\sum_{i=-j}^{j}\binom{k}{j-i}\binom{k-1}{j+i-1}
\end{split}
\end{equation*}
Evaluating each sum using the identity (\ref{vandermondeconvpos}) and
simplifying give
\begin{equation*}
\sum_{i=1}^{j}\binom{k}{j-i}\binom{k}{j+i}i^{2} = \frac{k}{4}\binom{2k-2}{2j-1}.
\end{equation*}
To obtain this formula we used the symmetry of the summands in
(\ref{chumomentspos}) under the substitution $i \to -i$.
\end{example}
Completely analogous considerations allow the computation of these
sums for higher values of $p$.

\section{An Explicit Example}
\label{ap:add_exs}

Many proofs in this article are based on a systematic approach
outlined in Proposition~\ref{delaysubleadingtwoprop} and further
discussed in Remark~\ref{linindepremark}.  The algebra involved is
rather cumbersome, but can be easily done with a symbolic algebra
computer package like those contained in Maple or Mathematica.
In order to give the reader a flavour of what such formulae look like,
in this appendix we report the explicit expressions that occur in the
computation of the second correction of the differences of the moments
of the transmission eigenvalues $\Delta \mathcal T^{(2,\delta)}_{k,2}(u)$.

In the proof of Proposition~\ref{unitarysubsubbeta2prop} we showed
that 
\begin{equation}
\label{eq:scrmte}
\Delta \mathcal{T}_{k,2}^{(2,\delta)}(u) = \frac{1}{k}
\sum_{j=0}^{k}\binom{k}{j}\binom{k}{j-1}\frac{u^{2k-2j}}{(u+1)^{2k+3}}
F_{k,2}^{(2,\delta)}(u,j).
\end{equation}
The coefficient $F_{k,2}^{(2,\delta)}(u,j)$ is polynomial of $2$-nd
order in $\delta$ and of $4$-th in $u$.  It is given explicitly by 
\begin{equation}
\label{p2tksubsub}
F^{(2,\delta)}_{k,2}(u,j) = A_{j}+\frac{\delta}{2}B_{j}+\frac{\delta^{2}}{4}C_{j},
\end{equation}
where
\begin{align*}
A_{j} &= \left( j-k \right)  \left( j-1-k \right)  \left( 3{j}^{2}-6j
k-j+k+3{k}^{2}-1 \right)/6 \\
&\quad +\left( j-1-k \right)  \left( 2j-1
-2k \right)  \left( j-k \right)u/3 \\
& \quad -j\left( j-1-k \right)\left( -jk+k+1+{j}^{2}-j \right)u^{2} \\
& \quad -j\left( 2j-1 \right)\left( j-1 \right) {u}^{3}/3
+j \left( j-1 \right)  \left( 3
{j}^{2}-5j+1 \right)u^{4}/6 ,\\
\frac{2B_{j}}{(u-1)} &=\left( j-k \right)  \left( j-1-k
 \right)  \left( 2j-1-2k \right)\\
  &\quad  -j(2j-1)\left( j-1 \right) {u}^{3} 
  +\left( 1+2j \right)\left( j-k \right)  \left( j-1-k
 \right)u \\
&\quad   -j \left(j-1\right)\left(2j-2k-3 \right)u^{2} \\
\intertext{and}
C_{j} &=   1/2 \left( j-k \right)\left( j-1-k \right)+u 
\left( k+1 \right)\left( j-1-k \right) \\
 &\quad +{u}^{2} \left( j \left( k-j+1 \right) 
   +k\bigl( k+1 \right)/2  \bigr) -
j \left( k+1 \right){u}^{3}+j\left( j-1 \right) {u}^{4}/2 .
\end{align*}

One of the main steps in the demonstration of
Proposition~\ref{unitarysubsubbeta2prop} was to turn
formula~\eqref{eq:scrmte} into a linear combination of Jacobi
polynomials. Below  we give the final expression. 

Define the transform of a sequence $\xi =(\xi_{l})_{l=0}^{\infty}$
(which may also depend on parameters $u$ and $k$) by the formula
\begin{equation*}
\mathcal{N}[\xi] := \frac{1}{k}\sum_{j=0}^{k}
\binom{k}{j}\binom{k}{j-1}\frac{u^{2k-2j}}{(u+1)^{2k+3}}\xi_{j}.
\end{equation*}
The second correction $\Delta \mathcal{T}_{k,2}^{(2,\delta)}(u)$ can be
written in terms of the sequences $A= (A_{l})_{l=0}^{\infty}$, $B =
(B_{l})_{l=0}^{\infty}$ and $C = (C_{l})_{l=0}^{\infty}$ as
\begin{equation*}
\Delta \mathcal{T}_{k,2}^{(2,\delta)}(u) = \mathcal{N}[A]+
\frac{\delta}{2}\mathcal{N}[B]+\frac{\delta^{2}}{4}\mathcal{N}[C].
\end{equation*}
By the definition of Jacobi polynomials (\ref{jacdef}) and minor
algebraic manipulations we arrive at
\begin{multline}
\label{NTcoeffa0}
(u+1)^{k+6}\mathcal{N}[A]=k(1-5k+3k^{2})
P^{(1,1)}_{k-2}(\tilde{u})(u-1)^{k-2}(u^{2}+1)u^{2}/6 \\
+k(k-1)(2-6k)P^{(2,1)}_{k-3}(\tilde{u})(u-1)^{k-3}u^{2}/6+(k-1)^{2}k
P^{(1,1)}_{k-3}(\tilde{u})(u-1)^{k-1}u^{2}(u+1)^{2}/2\\
-2k(k-1)P^{(1,0)}_{k-2}(\tilde{u})(u-1)^{k-2}u^{3}(u+1)/3+
kP^{(0,0)}_{k-1}(\tilde{u})(u-1)^{k-1}u^{2}(u+1)^{2}\\
-2k^{2}P^{(1,1)}_{k-2}(\tilde{u})(u-1)^{k-2}u^{3}(u+1)/3+2k(k-1)
P^{(1,2)}_{k-3}(\tilde{u})(u-1)^{k-3}u^{5}\bigl(1-u(3k-1)/2\bigr).
\end{multline}
The coefficient of $\frac{\delta}{2}$ is 
\begin{multline}
\label{NTcoeffa1}
(u+1)^{k+6}\mathcal{N}[B] = -k(k-1)P^{(1,0)}_{k-2}(\tilde{u})
(u-1)^{k-1}(u+1)u^{2} \\
+kP^{(1,1)}_{k-2}(\tilde{u})(u-1)^{k-1}(u+1)u^{2}+k(k-1)P^{(1,2)}_{k-3}
(\tilde{u})(u-1)^{k-2}u^{4}=0.
\end{multline}
This identity is obtained  using formula
(\ref{connection1}) and the three-term recurrence relation
(\ref{threetermrec}). The coefficient of $\frac{\delta^{2}}{4}$ is
given by
\begin{multline}
\label{NTcoeffa2}
(u+1)^{k+5}\mathcal{N}[C] = kP^{(1,1)}_{k-2}(\tilde{u})(u-1)^{k-2}
u^{2}(1+u^{2})/2\\
-(1+k)uP^{(1,0)}_{k-1}(\tilde{u})(u-1)^{k-1}(u+1)/2
+ku^{2}P^{(0,0)}_{k-1}(\tilde{u})(u-1)^{k-1}(u+1)\\
-(k+1)u^{3}P^{(0,1)}_{k-1}(\tilde{u})(u-1)^{k-1}(u+1)
+(k+1)u^{2}P^{(1,1)}_{k-1}(\tilde{u})(u-1)^{k-1}(u+1)/2.
\end{multline}
\bibliographystyle{amsplain}
\bibliography{chaotic_cav_bib}

\end{document}